%% file: main.tex
\documentclass[11pt]{article}
\usepackage[utf8]{inputenc} 
\usepackage{parskip}
\usepackage{amssymb, amsmath, amsthm, graphicx, subfigure}
\usepackage{enumerate}
\usepackage[dvipsnames]{xcolor}
\usepackage[colorlinks=true,linktoc=page]{hyperref}
\usepackage{braket}
\usepackage{esvect}
\usepackage[margin=1in]{geometry}
\usepackage{mathtools}
\usepackage{tikz}
\usepackage{algorithm}
\usepackage{algorithmicx}
\usepackage[noend]{algpseudocode}

\usepackage{mdframed}

\usetikzlibrary{patterns,patterns.meta}

\theoremstyle{plain}
\newtheorem{theorem}{Theorem}[section]
\newtheorem{lemma}[theorem]{Lemma}

\newtheorem{proposition}[theorem]{Proposition}

\theoremstyle{definition}

\theoremstyle{remark}

\newtheorem{example}[theorem]{Example}

\newcommand{\R}{\mathbb{R}}
\newcommand{\one}{\mathbf{1}}
\newcommand{\supp}{\operatorname{supp}}
\newcommand{\tr}{\operatorname{tr}}

\newcommand{\poly}{\operatorname{poly}}

\newcommand{\argmax}{\operatorname*{arg\,max}}

\newcommand{\ip}[2]{\left\langle #1,#2\right\rangle}
\newcommand{\norm}[1]{\left\lVert #1\right\rVert}

\newcommand{\Ex}{\mathbb{E}}
\newcommand{\Prb}{\mathbb{P}}

\newcommand{\E}{\mathbb{E}}

\title{Efficient Algorithms for Subdeterminant Maximization under Partition Matroids}

\author{Nikhil Bansal\thanks{University of Michigan. Email: \texttt{bansaln@umich.edu}. Supported by NSF awards CCF-2327011 and CCF-2504995.} \and Yuze Xu \thanks{University of Michigan. Email: \texttt{yuzex@umich.edu}}}

\begin{document}
\date{}
\maketitle

\begin{abstract}
\input{abstract}
\end{abstract}

\allowdisplaybreaks
\input{intro}

\input{geometric-program-no-highlights}

\input{continuous-dynamics-no-highlights}

\input{finite-steps-no-highlights}

\newpage
\appendix
\input{local-search-barrrier}

\input{terminal-appendix}
\input{small-k-no-highlights}

\bibliographystyle{alpha}  
\bibliography{ref}
\end{document}

%% file: abstract.tex
We consider the
determinant maximization problem under partition constraints: Given a PSD matrix $A \in \R^{n\times n}$ and a partition matroid $M$ on $[n]$, find a base $S$ of $M$ that maximizes $\det(A_{S,S})$.  
We give an $e^{O(k)}$-approximation algorithm to find such a set $S$, where $k$ is the rank of $M$.
This improves upon the current $k^{O(k)}$-approximation, and matches the current $e^k$-estimation guarantee, up to $O(1)$ factors in the exponent.
Our algorithm is based on rounding the geometric max-min relaxation due to \cite{NikolovSingh2016}, using a continuous potential-driven process, and several new structural and analytic properties of this relaxation.

%% file: intro.tex
\section{Introduction}
\label{sec:introduction}
In the constrained subdeterminant maximization problem, we are 
given a positive semidefinite (PSD) matrix $A \in \R^{n\times n}$, and we wish to find a principal submatrix $A_{S,S}$ maximizing $\det(A_{S,S})$, where the subset $S\subset [n]$ of rows and columns must satisfy certain natural combinatorial constraints. 

This framework models various problems arising in a range of areas, including statistics~\cite{Pukelsheim2006}, convex
geometry~\cite{Khachiyan1996}, 
 fair allocations \cite{AnariGharanSaberiSingh2016}, spectral graph theory \cite{NikolovSinghTantipongpipat2019}, network design and random processes \cite{KuleszaTaskar2012}.
The problem also has a rich theory and has been studied extensively in recent years, leading to several interesting results and connections
between areas such as optimization, convex analysis, matrix permanents, and geometry of polynomials \cite{AllenZhuLiSinghWang2017b, AnariGharanSaberiSingh2016, AnariGharan2017, AnariGharanVinzant2018, MadanNikolovSinghTantipongpipat2020, BrownLaddhaPittuSinghTetali2022}.

A natural family of constraints arises when the set $S$ must be a base of a given matroid $M$ on $[n]$. Let $k$ denote the rank of $M$.
The uniform matroid case (any \(|S|=k\) is feasible) was first studied by  Khachiyan~\cite{Khachiyan1996} who gave a \(k^{O(k)}\)-approximation. Later, in substantial progress, Nikolov~\cite{Nikolov2015} improved this to \(e^{k+o(k)}\). An approximation hardness of $c^k$ for some constant $c>1$ is also known~\cite{CivrilMagdonIsmail2013}.

In a seminal work, Nikolov and Singh \cite{NikolovSingh2016} considered the more general partition matroids. They gave a novel geometric max-min program, and showed, using stable-polynomial techniques, that it {\em estimates} the optimum value  within a factor $e^k$.
 After a series of works, this was eventually extended 
 to give an $e^{2k}$-estimation algorithm for {\em all} matroids by Anari, Oveis Gharan and Vinzant \cite{AnariGharanVinzant2018}, by combining stable-polynomial methods in \cite{StraszakVishnoi2017b, AnariGharan2017} with exciting advances on completely log-concave polynomials \cite{AnariGharanVinzant2018, AnariLiuGharanVinzant2019}. 

\medskip

{\bf Estimation vs.~algorithms.} Despite the remarkable progress, designing efficient algorithms that actually output such a set $S$ has proved considerably harder.
Ebrahimi, Straszak, and Vishnoi~\cite{EbrahimiStraszakVishnoi2017} gave the first $e^{O(n)}$ approximation for partition matroids and regular matroids, based on anti-concentration, but this does not give any guarantee in terms of $k$. The first bounds in terms of $k$ were given by \cite{MadanNikolovSinghTantipongpipat2020}, who developed sparsity results to give an $e^{O(k^2)}$ approximation for general matroid constraints. More recently, \cite{BrownLaddhaPittuSinghTetali2022} introduced a matroid-intersection approach based on exchange cycles to obtain a $k^{O(k)}$-approximation for general matroids.

Even for the special case of partition matroids, the best known approximation is $k^{O(k)}$ due to \cite{BrownLaddhaPittuSinghTetali2022}, leaving a substantial gap between the estimation and algorithmic bounds. 
Remarkably, Brown et al.~also show that the cycle-exchange methods cannot beat the $k^{O(k)}$ bound. We describe their instructive example in Appendix \ref{sec:local-barrier} for completeness.
In particular, designing an $\exp(O(k))$ approximation already for partition matroids is of significant interest.\footnote{We thank Mohit Singh for bringing this problem to our attention.}

\subsection{Our Result}
We resolve the question above, and show the following result.

\begin{theorem}
    \label{thm:main} Given a PSD matrix $A \in R^{n\times n}$ of rank $d$, and a partition  matroid $M$ on $[n]$ of rank $k\leq d$, there is a polynomial-time randomized algorithm that,  with probability at least $3/4$, returns a base $S$ of $M$ with $\det(A_{S,S})$ at least $\exp(-O(k))$ times the optimum.  
\end{theorem}

This matches the best estimation algorithm, up to $O(1)$ factors in the exponent.

Our approach differs completely from previous approaches and is based on continuous rounding methods. At a high level, we start with a solution to the max-min geometric relaxation of \cite{NikolovSingh2016} and modify it deterministically over time until the support has size $O(k)$, while decreasing the log-relaxation by at most $O(k)$. Once the support is $O(k)$, one can apply the algorithm of \cite{EbrahimiStraszakVishnoi2017} to obtain the final solution $S$.

Intuitively, it is the ability to modify the solution in continuous space that gives the algorithm much more flexibility and allows us to bypass the $k^{O(k)}$ barrier for exchange-based methods in \cite{BrownLaddhaPittuSinghTetali2022}.

We implement this idea as follows. First,  we determine how the objective of the max-min program changes under perturbations of the current solution $x$. Second, we consider a novel potential function to simultaneously track progress in the support reduction and the objective value. Based on this potential, we do a win-win argument: either there is a move with good first order improvement, then either the support is already small, or else there is move with good second order improvement. The max-min nature of the program can introduce degeneracies, but we show that in this case, we can make progress in a different way by splitting the problem into smaller instances. Third, we discretize the process to make it polynomial time. This poses several additional technical difficulties over the continuous analysis. A key novelty is a probabilistic interpretation of the derivatives based on the properties of the Cauchy-Binet distribution, and a Gibbs variational lower bound to get around the instability of the Hessian of the objective.

We now give a slightly more elaborate overview.

\subsection{Technical Overview}
Let $d$ be the rank of $A$. Then, we can write $A=V^\top V$, where $V\in R^{d\times n}$ with columns $v_1,\ldots,v_n \in \R^d$,  and for $k=d$, the objective is equivalent to finding $S$, subject to matroid constraints, to maximize \[\det(A_{S,S}) = \det \Big( \sum_{i\in S}v_i v_i^\top\Big).\]
For $k< d$, one considers the product of the $k$ non-zero eigenvalues of $\sum_{i\in S}v_i v_i^\top$. We only consider $k\leq d$, as any subdeterminant of $A$ has value $0$ for $k>d$.\footnote{The formulation $\det( \sum_{i\in S}v_i v_i^\top)$ however does makes sense for $k>d$, and has been extensively studied as the $D$-design problem. However, this is not our focus and we refer to \cite{lau0025a}  and references therein for more detail. In general, the $k>d$ setting is somewhat harder, and even the best known estimation guarantees under partition matroids are only $k^{O(k)}$ in this regime.}
Suppose for now that $k=d$ and that each part $P_a$ in the partition matroid has rank exactly $1$. The ideas for $k=d$ usually translate directly to $k< d$, and this will also be the case for us.

Nikolov and Singh proposed the following geometric max-min relaxation. There is a variable $x_i\in [0,1]$ for each $i\in [n]$, and a variable $y_a>0$ for each part $P_a$.
Consider the feasible region
$\Delta = \{x\geq 0: x(P_a)=1 \text{ for every } a\in [d]\}$, and let $\psi(x,y) =  \log \det (\sum_a y_a \sum_{i\in P_a} x_i v_iv_i^\top)$. 

Consider the program:
\[  \sup_{x\in \Delta} \inf_{\prod_a y_a=1} \psi(x,y).\]
It turns out this program can be solved in polynomial time. Note that the inner scaling $y=y(x)$ depends on $x$, and in general the infimum need not be attained if $x$ is degenerate.
 
Let us denote $F(x) = \inf_{\prod_a y_a=1} \psi(x,y)$. 

Our algorithm starts by computing an optimum solution $x^* = \argmax_{x\in \Delta} F(x)$ to the relaxation.

Starting from $x^*$, the algorithm will modify the solution over time, until the support is at most $O(k)$, while ensuring that $F(x) \geq F(x^*)-O(k)$. At this point one can apply the rounding of \cite{EbrahimiStraszakVishnoi2017}. For completeness, we give a self-contained exposition of this rounding for partition matroids.

There are three key ideas.

Suppose we are currently at a non-degenerate point $x$, so that $y(x)$ is attained.  
Now, as $x$ is perturbed to $\widetilde{x}=x+t\delta$ where $t$ is a tiny step size and $\delta$ is a unit direction, the inner minimizer $y(\widetilde{x})$ also changes.
First, we compute the infinitesimal dynamics of how $y(x)$ and $F(x)$ change. To obtain a polynomial-time algorithm later, we need suitable estimates with error bounds for these quantities, but let us ignore this for now. For a degenerate point $x$, where the inner minimizer may not exist, we use the structure of the partition matroid to show that the problem can be split into lower dimensional problems, without losing any objective values.

Second, we introduce a novel potential function $\Phi(x)$, based on the structure of the derivatives of $F$, and study the dynamics of
$L(x)=F(x)+\Phi(x)$.  The potential has positive curvature only on small coordinates, and is cleverly designed,  so that at each step either (i) there is a direction $\delta$ with large first order increase. Otherwise, it forces all leverage scores to be bounded. (ii) If all leverage scores are bounded, and if many
small coordinates remain, a dimension counting argument gives a direction $\delta$ that preserves the
partition constraints and kills the Schur-complement term coming from the inner
scaling minimization.  Along this direction the curvature of the
potential dominates the remaining negative curvature of $F$, and hence the second order increase of $L$ is large.

Third, we make the dynamics finite and polynomial time.  We never allow live
coordinates to become exponentially small: after a concave rebalancing step
certifies bounded leverage, coordinates below an inverse-polynomial threshold
are deleted one at a time.  Once all live coordinates have an inverse-polynomial
floor, the null-Schur perturbation can be taken with inverse-polynomial step
size.  This yields $O(d)$ support.  A final uniform-box probing and greedy
block-rounding argument converts the compressed fractional solution into an
integral transversal with another $\exp(O(d))$ loss.

\medskip{\bf Organization.} The rest of the paper is organized as follows. We describe the geometric program of \cite{NikolovSingh2016} in Section \ref{sec:geom-prog}. We also give some useful properties of this program. In particular, in Section \ref{sec:derivatives} we obtain expressions for the various derivatives of the objective, and in Section \ref{sec:nondegenerate-supports} we describe how the instance can be split into smaller instances upon reaching a degenerate point $x$. 
For ease of exposition, in Section \ref{sec:continuous-compression}, we first describe the continuous dynamics, ignoring the issue of convergence and polynomial number of steps.  We also give a simple self-contained exposition of the rounding algorithm of \cite{EbrahimiStraszakVishnoi2017} for completeness. 
In Section \ref{sec:finite-dynamics}, we describe the discretized dynamics and the final polynomial-time algorithm. To this end, we describe a probabilistic interpretation of the derivatives of $F(x)$ and its various properties. Both Sections \ref{sec:continuous-compression} and \ref{sec:finite-dynamics}  consider only the case of $k=d$.
In Appendix \ref{sec:k-le-d} we describe the minor changes needed to handle the case of $k\leq d$. 

\paragraph{AI Use.}
The authors used GPT-5.5 Pro during the development of this work to explore proof strategies, search for related literature, and assist with verification. GPT was not used in any part of the exposition.

%% file: geometric-program-no-highlights.tex
\section{The Geometric Program}
\label{sec:geom-prog}
We begin by defining the problem formally.

We are given vectors $v_1,\ldots,v_n \in \R^d$ and a partition $P_1,\ldots,P_\ell$ of $[n]$, with integers $b_a\geq 1$ for each part $a\in [\ell]$. We assume that $k:=\sum_a b_a \leq d$.
The goal is to find a subset $B\subseteq[n]$ with $|B\cap P_a|=b_a$ for each part $a\in [\ell]$ that maximizes the squared $k$-dimensional volume of the parallelotope formed by the chosen vectors $v_i$ for $i\in B$.

Let $V_B$ be the $d\times k$ matrix with columns $v_i$ for $i\in B$.
The objective is
        $\det(V_B^\top V_B)$. 
When $k=d$, this is equivalently
\[
        \det(V_B)^2
        = \det(V_BV_B^\top)
        = \det \Big(\sum_{i\in B} v_iv_i^\top\Big).
\]
When $k<d$, it is still $\det(V_B^\top V_B)$; equivalently it is the $k$-th elementary symmetric polynomial of $\sum_{i\in B}v_iv_i^\top$. This equals the product of the $k$ nonzero eigenvalues when $V_B$ has full column rank, and is zero otherwise.
Here, we focus on the case of $k=d$ --- the case of $k<d$ is described in Section~\ref{sec:k-le-d}.
Without loss of generality, we can also assume that each part has rank $1$.\footnote{Replace a part $P_a$ of capacity $b_a$ by $b_a$ identical rank-one parts, each containing the same vectors from $P_a$. Any transversal of the copied instance that chooses the same original vector twice has zero $k$-volume, since the corresponding columns are linearly dependent. Thus, for positive-volume solutions, this reduction is exact. If the optimum volume is zero, the problem is trivial.}

Thus we focus on the following problem:
Given $v_1,\ldots,v_n \in \R^d$ and a partition $\{P_1,\ldots,P_d\}$ of $[n]$, find a \emph{transversal} $B$ with $|B\cap P_a|=1$ for each $a$, maximizing
\[
        \det\Big(\sum_{i \in B} v_i v_i^\top\Big).
\]
Our algorithm will recurse on residual instances with $r$-dimensional vectors and $r$ parts, where $r\leq d$. For this reason, we will use $r$ for the number of parts instead of $d$, and reserve $d$ for the original ambient dimension.

\subsection{Naive Convex Relaxation}
Consider a residual rank-$r$ instance. 
Using the decision variables $x_i\in \{0,1\}$ for each vector $v_i$, we have the following exact formulation (for the logarithm of the original objective).
\[
    \max_x\,\, \log \det \Big(\sum_i x_i v_i v_i^\top\Big), \qquad \text{ s.t. $x(P_a)=1$ for all $a \in [r]$, and  $x \in \{0,1\}^n$}.
\]
Let us define the function $A(x) := \sum_i x_i v_iv_i^\top$.
Recall that $\log \det (\cdot)$ is concave on the positive definite cone, and $A(x)$ is affine in $x$. Thus $x\mapsto \log\det A(x)$ is concave on the region where $A(x)\succ0$. Consider the polytope 
\[\Delta_r:=\{x\ge0: x(P_a)=1\text{ for every }a\in[r]\}.\]
Then we have the natural convex programming relaxation:
    $\max_{x \in \Delta_r} \,\, \log \det A(x)$.
\smallskip

Unfortunately, this program can be arbitrarily bad. To get around this, Nikolov and Singh \cite{NikolovSingh2016} proposed a very interesting geometric program.
As we will be using this program and its properties extensively, it is instructive to see a simple bad example for the naive program, and a fix for it.

\smallskip

\begin{example}
    Let $r=2$, and let $P_1=\{Me_1,Me_2\}$ and $P_2 = \{e_1,e_2\}$, where $M$ is arbitrarily large.
Clearly, $\det(V_B)\leq M$ for any integral transversal $B$ and thus 
$\det A(x) \leq M^2$ for integral $x$. On the other hand, for $x=(1/2,1/2,1/2,1/2)$ we have $A(x) = ((M^2+1)/2)I$, and thus  $\det A(x) = \Omega(M^4)$.  
\end{example} 
 Nikolov and Singh \cite{NikolovSingh2016} proposed the following fix. 
Suppose we scale the outer products in $P_1$ by $y_1=1/M$ and in $P_2$ by $y_2=M$. Then the determinant of every integral transversal is unchanged, because $y_1y_2=1$. On the other hand, after this rescaling we have $A(x) \preceq 2MI$ for every fractional $x$, eliminating the gap in the example above.

\subsection{Saddle Geometric Relaxation} 
For each part $P_a$, introduce a \emph{scaling} variable $y_a>0$. Let us write $A_a(x):=\sum_{i\in P_a}x_i v_i v_i^\top$.
Nikolov and Singh proposed the following relaxation.
\begin{equation}
\label{eq:prog-a}
    \max_{x\in \Delta_r}\,\, \inf_{y: \prod_a y_a=1} \,\,\, \log \det \Big(\sum_a y_a A_a(x)\Big) 
\end{equation}
As noted above, the scaling condition $\prod_a y_a=1$ is invisible to
integral transversals, and
thus this is a valid relaxation.

Let us write $y_a = e^{z_a}$ for $z_a \in \R$.
The condition $\prod_a y_a=1$ becomes $z \in \one^\perp :=\{z\in\R^r:\sum_a z_a=0\}$.
For $z \in \R^r$, let us define the functions      
\begin{equation}
\label{eq:psi-definition}
        M(x,z):=\sum_{a=1}^r e^{z_a}A_a(x), \qquad \psi(x,z):=\log\det M(x,z).
\end{equation}
{\bf Reformulation in $z$-variables.} Thus the program \eqref{eq:prog-a} can now be written as
\begin{equation}
\label{eq:prog-b}
    \sup_{x\in \Delta_r} \inf_{z \in \one^\perp} \,\,\, \psi(x,z)
\end{equation}
The function $\psi(x,z)$ has nice properties. It is concave in $x$ and convex in $z$. The first follows as $\log \det (M)$ is concave on the PSD cone, and $M$ is affine in $x$. The convexity in $z$ follows from the Cauchy--Binet expansion, which writes $\psi(x,z)$ as a log-sum-exp of affine functions of $z$. In Lemma~\ref{lem:derivatives-psi}, we will explicitly compute the $z$-Hessian of $\psi(x,z)$ and show that it is positive semidefinite (PSD). 
This saddle structure allows the program \eqref{eq:prog-b} to be solved efficiently, to arbitrary accuracy, using standard convex--concave optimization techniques.

Notice that given any feasible $x \in \Delta_r$, the inner minimizer, when it exists, is a function $z=z(x) \in \one^\perp$ of $x$.
In general, the inner infimum need not always be attained for certain \emph{degenerate} $x$. However, such degenerate $x$ will have further structure that our algorithm will exploit, and we will discuss this in detail in Section~\ref{sec:nondegenerate-supports}. Given $x \in \Delta_r$, let us define
\begin{equation}
\label{eq:F-definition}        F(x):=\inf_{z\in\one^\perp}\psi(x,z).
\end{equation}
Let $x^*$ denote an optimal solution to $\sup_{x\in \Delta_r} F(x)$; algorithmically, we compute a sufficiently accurate approximate saddle point.



\subsection{Derivatives of $\psi(x,z)$ and  $F(x)$}
\label{sec:derivatives}
Our goal is to understand how the value $F(x)$ changes when we perturb the point $x$.
We assume throughout this section that $x$ is nondegenerate so that the inner-minimizer $z(x)$ is well-defined and unique. 

Fix a point $(x,z)$ with $M:=M(x,z)\succ 0$.
Note that we treat $x$ and $z$ as independent variables here.
Let  $w_i:=e^{z_{a(i)}/2}v_i$, where $a(i)$ is the part containing $i$. Then we have that
        \[M=   \sum_a e^{z_a} \sum_{i\in P_a} x_i v_i v_i^\top = \sum_i x_i w_iw_i^\top.\]
We now define some related quantities that will be useful to express the derivatives. Let $Q_i = M^{-1/2}w_iw_i^\top M^{-1/2}$ and,
\[\ell_i:= w_i^\top M^{-1} w_i = \tr( w_i w_i^\top M^{-1}) = \tr Q_i.\] 
For part $a$, let $M_a:=\sum_{i\in P_a}x_iw_iw_i^\top$ denote the contribution of part $P_a$ to $M$, and define  
\[B_a:=M^{-1/2}M_a M^{-1/2}.\]
Note that  $\sum_a B_a = I$ and $\tr(B_a) = \tr( M_a M^{-1}) = \sum_{i\in P_a} x_i \ell_i$. 

Notice that $\partial_{x_i}M =w_iw_i^T$. Therefore, for any tangent direction $\delta$ for $x$, the differential 
$D_x M[\delta] = \sum_i \delta_i w_iw_i^\top$. We define 
\[C_a(\delta):=M^{-1/2}(D_x M_a[\delta])M^{-1/2} = \sum_{i\in P_a}\delta_i Q_i,
\]
and let $C(\delta) := \sum_a C_a(\delta)= \sum_{i}\delta_i Q_i$.

\paragraph{Derivatives of $\psi(x,z)$.}
\begin{lemma} The derivatives of $\psi(x,z)$ satisfy the following.
\label{lem:derivatives-psi}
\begin{enumerate}
    \item ($x$-derivative). For every direction $\delta$ for $x$,
\[
        D_x\psi(x,z)[\delta]=\tr(C(\delta)), \qquad  D^2_{xx}\psi(x,z)[\delta,\delta]
        =-\tr\big(C(\delta)^2\big).\]
\item ($z$-derivatives).
For every $z$-direction $\zeta$,
        \[D_z \psi(x,z)[\zeta]= \sum_a \zeta_a \tr B_a, \qquad  
D^2_{zz}\psi(x,z)[\zeta,\zeta] = \zeta^\top \Gamma \zeta\]
where the $z$-Hessian $\Gamma$ has entries 
$\Gamma_{ab} =\mathbf{1}_{a=b}\tr B_a-\tr(B_aB_b)$.
In particular, $\Gamma \succeq 0$.
\item (Mixed derivative).
For every $x$-direction $\delta$ and every part $a$,
\[
        D_x\left(\frac{\partial\psi}{\partial z_a}\right)[\delta]
        =
        \tr C_a(\delta)-\tr\big(B_aC(\delta)\big).
\]
Define the vector $\eta(\delta) \in \R^r$  with coordinates $\eta_a(\delta) =  \tr C_a(\delta)-\tr\big(B_aC(\delta)\big)$. Then $\eta(\delta) \perp \one$.
\end{enumerate}
\end{lemma}

\begin{proof}
We will repeatedly use the following standard identities for any invertible matrix $M$, 
\begin{equation}
    \label{eq:basic-derivatives}
    d(\log \det M) = \tr(M^{-1}dM), \qquad d(M^{-1}) = -M^{-1}(dM)M^{-1}.
\end{equation}
Let us denote $\dot M := D_x M[\delta] = \sum_i \delta_i w_iw_i^\top$.
Using the first identity in \eqref{eq:basic-derivatives}
\begin{align*}
        D_x\psi(x,z)[\delta]
        &=\tr(M^{-1} \dot M) = \tr (C(\delta)).
\end{align*}
Differentiating again and noting that $D_{xx} M[\delta,\delta]=0$ as $M$ is affine in
$x$, and using the second identity in \eqref{eq:basic-derivatives} gives
\begin{align*}
        D^2_{xx}\psi(x,z)[\delta,\delta]
        &=-\tr(M^{-1} \dot M M^{-1} \dot M)
        =-\tr\big(C(\delta)^2\big).
\end{align*}
This gives the $x$-derivatives, and we now consider the $z$-derivatives.

As $M = \sum_a e^{z_a} \sum_{i\in P_a} x_i v_i v_i^\top$, we have that $\partial_{z_a}M= e^{z_a}\sum_{i\in P_a} x_i v_i v_i^\top = M_a$, and therefore
\begin{align}
\label{eq:partial-psi-z}
        \frac{\partial\psi}{\partial z_a} =\tr(M^{-1}M_a) 
        =\tr B_a.
\end{align}
For the second derivative, differentiating the expression in \eqref{eq:partial-psi-z} once more and using that
        $\partial_{z_b}M^{-1}=-M^{-1}M_bM^{-1}$ and that $
\partial_{z_b}M_a=\mathbf{1}_{a=b}M_a$
we have that,
\begin{align*}
    \Gamma_{ab}:=    \frac{\partial^2\psi}{\partial z_a\partial z_b}
        &=
        \tr\left((\partial_{z_b}M^{-1})M_a\right)
        +\tr\left(M^{-1}\partial_{z_b}M_a\right) \\
        &=-\tr(M^{-1}M_bM^{-1}M_a)
        +\mathbf{1}_{a=b}\tr(M^{-1}M_a) \\
        &=-\tr(B_bB_a)+\mathbf{1}_{a=b}\tr B_a.
\end{align*}
This $z$-Hessian $\Gamma$ is PSD as,
\begin{align*}
     D^2_{zz}\psi(x,z)[\zeta,\zeta] &=    \sum_a\zeta_a^2\tr B_a
        -\sum_{a,b}\zeta_a\zeta_b\tr(B_aB_b) \\
        &=\frac12\sum_{a,b}\tr(B_aB_b)(\zeta_a-\zeta_b)^2 \geq 0,
\end{align*}
where in the second line we used that $\sum_a \zeta_a^2 \tr(B_a) = \sum_{a,b} \zeta_a^2 \tr(B_aB_b)$ as $\sum_bB_b=I$. 
The non-negativity follows as $\tr(B_aB_b)\geq 0$ as $B_a,B_b\succeq0$.

Finally, we compute the mixed derivative. 

By \eqref{eq:partial-psi-z} we have that $\partial\psi/\partial z_a=\tr(M^{-1}M_a)$. Taking differential of this along the $x$-direction $\delta$ gives,
\begin{align*}
        \eta_a(\delta):= D_x\left(\frac{\partial\psi}{\partial z_a}\right)[\delta]
        &=
        \tr(M^{-1} (D_x M_a[\delta]))
        -\tr(M^{-1} (D_x M[\delta]) M^{-1}M_a) \\
        &=
        \tr C_a(\delta)-\tr\big(C(\delta)B_a\big).
\end{align*}
To see that $\eta(\delta)\in\one^\perp$,
observe that
\[
\sum_a \eta_a(\delta)
=\sum_a \bigl(\tr C_a(\delta)-\tr(C(\delta)B_a)\bigr)
=\tr(C(\delta))-\tr(C(\delta))=0,
\] where we use that $\sum_a C_a(\delta)=C(\delta)$ and $\sum_a B_a =I$.
\end{proof}

\paragraph{Derivatives of $F(x)$.}
\label{sec:F-calculus}
We now describe how
$F(x)=\inf_{z\in\one^\perp}\psi(x,z)$ changes when $x$ is perturbed infinitesimally. To do this, we will explicitly describe $z(x)$ as a function of $x$, and then use the results in Lemma~\ref{lem:derivatives-psi}.
 
We assume throughout this section that $x$ is
nondegenerate and the inner
infimum is attained by a unique minimizer $z(x)\in\one^\perp$. Degenerate points are handled by the exact splits in Section~\ref{sec:nondegenerate-supports}.
All the  quantities such as $w_i,\ell_i,B_a,C_a(\delta),C(\delta)$ and $\eta(\delta)$ are henceforth evaluated at
$z=z(x)$.


We first describe some properties of $\psi(x,z)$ at the minimizer $z(x)$.
\begin{lemma}
\label{lem:partwise-balance}
For any $x$ for which the inner minimizer $z(x)$ exists, 
the $z$-gradient $\nabla_z \psi(x,z) = \one$, and hence $\tr B_a=1$ for each part $a\in[r]$.
Moreover, for any point $(x,z(x))$, the vector $\one$ lies in the kernel of the $z$-Hessian $\Gamma$ of $\psi(x,z)$.
\end{lemma}
\begin{proof}
Fix an $x\in \Delta_r$. As the only constraint on $z$ is that $z\in \one^\perp$, taking the Lagrangian $\psi(x,z) +\lambda \langle z,\one \rangle$ and setting the gradient with respect to $z$ gives that $\nabla_z \psi(x,z) = -\lambda \one$ for some $\lambda \in \R$.

By Lemma~\ref{lem:derivatives-psi}, the $z$-gradient has coordinates $(\nabla_z \psi(x,z))_a  = \tr(B_a)$, and thus all $\tr(B_a)$ are equal.
As $\sum_aB_a=I$ we have that 
        $\sum_a\tr B_a=r$, and thus each $\tr(B_a)=1$.
  
 The second claim follows as the $z$-Hessian has entries
$\Gamma_{ab}=\mathbf{1}_{a=b}-\tr(B_aB_b)$, and thus for each $a$, the entry  \[(\Gamma \one)_a = \sum_b \Gamma_{ab} = 1 - \sum_b \tr(B_aB_b) = 1-\tr(B_a)= 0,\] 
where the third equality uses that $\sum_b B_b=I$ and the last equality uses that $\tr(B_a)=1$.
\end{proof}

We now compute the derivatives of $F(x)$. This will be the main tool used in Section~\ref{sec:continuous-compression}.
Recall the definition of the vector $\eta(\delta)$ of mixed derivatives, from Lemma~\ref{lem:derivatives-psi}. We have the following. 
\begin{lemma}[Derivatives of $F$]
\label{lem:F-derivatives}
At a nondegenerate point $x$, for every feasible direction $\delta$, 
\begin{equation}
\label{eq:F-hessian}
       DF(x)[\delta]=\tr(C(\delta)) \quad \text{ and } \quad D^2F(x)[\delta,\delta]
        =
        -\tr\big(C(\delta)^2\big)
        -\eta(\delta)^\top\Gamma_{\one^\perp}^{-1}\eta(\delta).
\end{equation}
Here $\Gamma_{\one^\perp}^{-1}$ denotes the inverse of $\Gamma$ restricted to
$\one^\perp$. 
In particular, if $\eta(\delta)=0$, then
        \[D^2F(x)[\delta,\delta]
        =-\tr\big(C(\delta)^2\big).\]
\end{lemma}

\begin{proof}
Fix a point $x$, and let $z=z(x)$ be the inner minimizer. When $x$ moves infinitesimally in the direction $\delta$, let
\[
        \zeta:=Dz(x)[\delta]
\]
denote the direction of movement of $z(x)$. 
As $z(x) \in \one^\perp$ for every $x$, we have that $\zeta \in\one^\perp$.
Then, by the chain-rule
\begin{align*}
        DF(x)[\delta]&=D_x\psi(x,z)[\delta] + D_z\psi(x,z)[\zeta] \\
        &=D_x\psi(x,z)[\delta] = \tr(C(\delta)).
\end{align*}
Here, the second step uses that $D_z\psi(x,z)[\zeta] =0$ as $\nabla_z\psi(x,z) = \one$ at $z=z(x)$ by Lemma~\ref{lem:partwise-balance},  and that $\zeta \perp \one$. The last equality follows by Lemma~\ref{lem:derivatives-psi}.

To obtain $\zeta$ explicitly, we differentiate the condition from Lemma~\ref{lem:partwise-balance}  that $\Pi_{\one^\perp}\nabla_z\psi(x,z)=0$ at $z=z(x)$,
 in the $x$-direction $\delta$. Applying the chain-rule, and using the expressions in  Lemma~\ref{lem:derivatives-psi}, this gives
\[
        \Pi_{\one^\perp}D_x\nabla_z\psi(x,z)[\delta]+\Gamma\zeta=0.
\]
Since $\eta(\delta)\in\one^\perp$, this is equivalently $\eta(\delta)+\Gamma\zeta=0$ inside $\one^\perp$. As $\one$ lies in the kernel of $\Gamma$ by Lemma~\ref{lem:partwise-balance} and $\zeta \in \one^\perp$, this gives  \[\zeta=-\Gamma_{\one^\perp}^{-1}\eta(\delta).\]
For the second derivative of $F$, we differentiate $F(x)=\psi(x,z(x))$ twice along $\delta$.  By the chain rule
\begin{align}
        D^2F(x)[\delta,\delta]
        =D^2_{xx}\psi[\delta,\delta]
        +2D^2_{xz}\psi[\delta,\zeta]
        +D^2_{zz}\psi[\zeta,\zeta].
\label{eq:d2f}
\end{align}
By Lemma~\ref{lem:derivatives-psi}, these terms are  
\[
        D^2_{xx}\psi[\delta,\delta]
        =-\tr(C(\delta)^2), \quad 
        D^2_{xz}\psi[\delta,\zeta]
        =\ip{\eta(\delta)}{\zeta}, \text{ and } D^2_{zz}\psi[\zeta,\zeta]=\zeta^\top\Gamma\zeta.
\]
Using $\Gamma\zeta=-\eta(\delta)$ (inside $\one^\perp$), the last two terms in \eqref{eq:d2f} give,
\begin{align*}
        2\ip{\eta(\delta)}{\zeta}+\zeta^\top\Gamma\zeta
        &=2\ip{\eta(\delta)}{\zeta}-\ip{\eta(\delta)}{\zeta} 
        =\ip{\eta(\delta)}{\zeta} =-\eta(\delta)^\top\Gamma_{\one^\perp}^{-1}\eta(\delta).
\end{align*}
Together with the first term in \eqref{eq:d2f}, which equals $-\tr(C(\delta)^2)$, gives the claimed result in \eqref{eq:F-hessian}.
\end{proof}

\subsection{Nondegenerate supports}
\label{sec:nondegenerate-supports}
Fix a rank-$r$ instance and a point $x\in\Delta_r$.  Let us denote the support
$
E:=\supp(x)=\{i:x_i>0\}$ and let $
E_a:=E\cap P_a$ for each part $a\in [r]$.
For a subset of parts $S\subseteq[r]$, define
\[
\rho_E(S)
:=
\dim \operatorname{span}\{v_i:i\in E_a,\ a\in S\}.\]
Note that $\rho_E(S)$ depends only on the support $E$, and not on the specific positive values of the coordinates $x_i$.
We say that the support $E$ is \emph{admissible} if
$\rho_E([r])=r$ and
$\rho_E(S)\ge |S|$, 
for all  $S\subseteq[r]$. By Rado's theorem (the linear Hall theorem), any admissible $E$
 contains a full-rank transversal.
Thus, given an $x$ with support $E$, the objective $F(x)$ is finite if and only if $E$ is admissible.

{\bf Nondegenerate support.}
A proper nonempty set $S\subsetneq[r]$ is called \emph{tight} if $
\rho_E(S)=|S|$.
A support
$E$ is called \emph{nondegenerate} if it is admissible and has no proper tight set, i.e.,
\[
\rho_E(S)>|S|
\qquad
\text{for all } \emptyset\ne S\subsetneq[r].\]
We have the following important fact.

\begin{proposition}
\label{prop:finiteness-balanced-scaling}
Let $x\in\Delta_r$ and $E=\supp(x)$. The inner infimum defining $F(x)$ is attained by a unique minimizer $z(x)\in\one^\perp$ iff $E$ is nondegenerate.
Moreover, the $z$-Hessian at $z(x)$ is positive definite on $\one^\perp$.
\end{proposition}
We also note a quantitative bound that will motivate our polynomial-time algorithm.
\begin{proposition}
\label{lem:scaling-box}
Suppose $E=\supp(x)$ is nondegenerate, and every coordinate in $E$ satisfies $x_i\ge \alpha>0$. Then
$
\|z(x)\|_2
\le
\sqrt{r} \log(C(x)/c_{\min}(x)),$
where, for $r$-subsets $T\subseteq E$, $c_T(x):=\det(V_T)^2\prod_{i\in T}x_i$ are the Cauchy--Binet coefficients,
\[
        C(x):=\sum_{T\subseteq E,\ |T|=r}c_T(x),
        \qquad
        c_{\min}(x):=\min_{T:c_T(x)>0}c_T(x).
\] 

In particular, $\|z(x)\|_2$ is polynomial in $r, \log (1/\alpha)$ and the bit complexity of the input vectors.    
\end{proposition}
We refer to \cite{MadanNikolovSinghTantipongpipat2020} (Appendix B) for the proof of these results.

\subsection{Exact tight-set split}
Suppose $x$ is a degenerate point. In this case, we can split the instance into lower-dimensional instances without any loss in objective value. More precisely, one has the following. 

\begin{lemma}[Tight-set split]
\label{lem:tight-split-pedagogical}
Given $x$, suppose $\emptyset\ne S\subsetneq [r]$ is tight with
$\rho_E(S)=|S|$. Let
$W:=\operatorname{span}\{v_i:i\in E_a,\ a\in S\}$ and let $\Pi$ be the orthogonal projection onto $W^\perp$. Then, the instance decomposes into two disjoint instances, one with parts in
$S$ and the original vectors, viewed in $W$, and the other with parts in $S^c$
and projected vectors $\Pi v_i$, viewed in $W^\perp$. Moreover, $x$ decomposes into $x_{|S}$ and $x_{|S^c}$, and
\[
F(x)=F^W(x_{|S})+F^{W^\perp}(x_{|S^c}),
\]
where the superscript indicates that the same scaling relaxation is evaluated in the indicated subspace.
Similarly, for every integral transversal $B$ supported on $E$, with $B=B_S\cup B_{S^c}$, the determinant decomposes as
\[
\det\Big(\sum_{i\in B}v_i v_i^\top\Big)
=
\det_W\Big(\sum_{i\in B_S}v_i v_i^\top\Big)
\det_{W^\perp}\Big(\sum_{i\in B_{S^c}}(\Pi v_i)(\Pi v_i)^\top\Big).
\]
\end{lemma}

\begin{proof}
Use the orthogonal decomposition
$\R^r=W\oplus W^\perp$.
The parts in $S$ lie entirely inside $W$, and the parts in $S^c$ may have components in both $W$ and $W^\perp$.
So we can write
\begin{equation}
\label{eq:split-1}
\sum_{a\in S}y_aA_a(x)=
\begin{pmatrix}
P(y_S)&0\\
0&0
\end{pmatrix}, \qquad 
\sum_{b\notin S}y_bA_b(x)=
\begin{pmatrix}
N_{11}(y_{S^c})&N_{12}(y_{S^c})\\
N_{21}(y_{S^c})&N_{22}(y_{S^c})
\end{pmatrix}.
\end{equation}
Here $N_{22}$ is the matrix formed from the projected outer vectors $\Pi v_i$ in
$W^\perp$. So the full matrix 
\begin{equation}
\label{eq:split-2}
M(y)=
\begin{pmatrix}
P(y_S)+N_{11}(y_{S^c})&N_{12}(y_{S^c})\\
N_{21}(y_{S^c})&N_{22}(y_{S^c})
\end{pmatrix}.
\end{equation}

First we prove 
$F(x)\ge F^W(x_{|S})+F^{W^\perp}(x_{|S^c})$.
For ease of notation, let us drop $y_S,y_{S^c}$ in \eqref{eq:split-1}, \eqref{eq:split-2}.
By the Schur complement formula,
\[
\det M(y)=\det(N_{22})
\det\bigl(P+N_{11}-N_{12}N_{22}^{-1}N_{21}\bigr) \geq \det(N_{22})\det(P).
\]
Here the inequality uses that the matrix in \eqref{eq:split-1}, corresponding to $S^c$ 
is PSD, and hence its Schur complement
$N_{11}-N_{12}N_{22}^{-1}N_{21} \succeq 0$ (if $N_{22}$ is singular, we can add $\epsilon I$ and take the limit).

We now give a valid scaling for $S,S^c$. Consider the original scaling $y$, and let
$
y_S:=\prod_{a\in S}y_a$ and $
y_{S^c}:=\prod_{b\notin S}y_b$. Note that $y_Sy_{S^c}=1$. 
Normalize the scaling inside $S$ by dividing each $y_a$ by $y_S^{1/m}$, where $m=|S|$, and normalize the scaling inside $S^c$ analogously. Then $\det P$ changes by the factor $y_S$ and $\det N_{22}$ changes by the factor $y_{S^c}$. Since $y_Sy_{S^c}=1$, the product $\det(N_{22})\det(P)$ is unchanged. Minimizing over the two normalized child scalings can only decrease the right-hand side, and hence $F(x)\ge F^W(x_{|S})+F^{W^\perp}(x_{|S^c})$.

Now we prove the reverse inequality $F(x)\le F^W(x_{|S})+F^{W^\perp}(x_{|S^c})$.  
Fix unit-product scalings for the two subinstances:
$\prod_{a\in S}\widehat y_a=1$, and 
$\prod_{b\notin S}\widehat z_b=1$.
Let us denote $m=|S|$. 

Let $t>0$ be a parameter. We combine these into a full product-one scaling as 
\[
 y_a(t)=t^{r-m}\widehat y_a\quad(a\in S),
 \qquad
 y_b(t)=t^{-m}\widehat z_b\quad(b\notin S).
\]
This is valid as
$\prod_{a\in S}t^{r-m}\prod_{b\notin S}t^{-m}
=t^{m(r-m)}t^{-m(r-m)}=1$.
Under this scaling,
\[
M(t)=
\begin{pmatrix}
 t^{r-m}P+t^{-m}N_{11} & t^{-m}N_{12}\\
 t^{-m}N_{21} & t^{-m}N_{22}
\end{pmatrix},
\]
where $P,N_{11},N_{12},N_{21},N_{22}$ are computed from the fixed child scalings
$\widehat y$ and $\widehat z$ on the two child instances.

Consider the diagonal matrix $D_t = t^{-(r-m)/2}I_W \oplus t^{m/2}I_{W^\perp}$.
As $\dim W=m$ and $\dim W^\perp=r-m$, we have
$\det D_t=t^{-m(r-m)/2}t^{m(r-m)/2}=1
$, and 
thus
$
\det M(t)=\det(D_tM(t)D_t)$.
But
\[
D_tM(t)D_t=
\begin{pmatrix}
 P+t^{-r}N_{11} & t^{-r/2}N_{12}\\
 t^{-r/2}N_{21} & N_{22}
\end{pmatrix}.
\]
Letting $t\to\infty$, the off-diagonal blocks and the extra $N_{11}$ term vanish
and we get that
$
\det M(t)\rightarrow \det_W(P)\det_{W^\perp}(N_{22})$.
Taking infima over the two child scalings gives
$
F(x)\le F^W(x_{|S})+F^{W^\perp}(x_{|S^c})$.

The integral determinant identity follows as the full parallelotope volume factors as the base volume in $W$ times  the height volume after projection to $W^\perp$.
\end{proof}
Applying this split repeatedly, we can ensure that the restriction of $x$ to every residual instance is nondegenerate.

\paragraph{Singleton part contraction.} 
A useful special case occurs when $x_i=1$ for some coordinate $i$. In this case, the part $P_a$ containing $i$ becomes degenerate. 
We record the corresponding contraction separately.

\begin{lemma}
\label{lem:singleton-pedagogical}
Suppose a part $P$ has a single vector $a:=v_i$ with
$x_i=1$.  Let $\Pi$ be the orthogonal projection onto $a^\perp$. Then for
the residual instance with projected vectors $\Pi v_i$ in the remaining parts, one has
\[
F(x)=\log\|a\|_2^2+F^{a^\perp}(x'),
\]
where $x'$ is $x$ restricted to remaining parts.
Moreover, for every transversal $T$ of remaining parts,
\[
\det\Big(aa^\top+\sum_{i\in T}v_i v_i^\top\Big)
=
\|a\|_2^2\,\,
\det_{a^\perp}\Big(\sum_{i\in T}(\Pi v_i)(\Pi v_i)^\top\Big).
\]
\end{lemma}

The important point is that this will allow us to apply the following three structural steps without any loss in objective value: (i) delete coordinates with $x_i=0$; (ii) contract singleton parts; (iii) split along proper tight sets.

Note that these steps only depend on the support and not on the specific values of $x$. Thus zero-coordinate deletions are performed at most 
$O(n)$ times and the split and singleton contractions are performed only $O(d)$ times throughout the algorithm.

Whenever the support changes, the tight sets can also be determined algorithmically. To do this we observe that $g(S)= \rho(S)-|S|$ is a submodular function, because $\rho$ is a linear-matroid rank function and $|S|$ is modular.
Thus tight or violated Hall-type inequalities can be found by submodular minimization using standard reductions.
More simply, one can use matroid intersection between the represented linear matroid and
the partition matroid.

%% file: continuous-dynamics-no-highlights.tex
\section{Continuous Process}
\label{sec:continuous-compression}
We now describe the continuous process for reducing the support to $O(r)$. We assume throughout that we are at a nondegenerate point $x$, and show the existence of a local direction $\delta$ at $x$, such that the process makes large progress for the objective $L(x)=F(x)+\Phi(x)$, as long as the support is large. If $x$ is at a degenerate point, we apply the split discussed previously and proceed on the lower-dimensional instances.

Notice that we are completely ignoring any convergence or running-time issues. In particular, the continuous process may not actually converge to a degenerate point in a finite number of steps, since it may be forced to take smaller and smaller steps as the point $x$ becomes more degenerate.
However, this continuous description allows us to describe several key ideas involving the potential function $\Phi$ and the dimension counting argument for the null-Schur direction, in a clean way. The full discrete version is described in Section \ref{sec:finite-dynamics}. 

We begin by defining the potential function $\Phi$.

\subsection{The Potential Function}
\label{subsec:capped-potential}

Fix the constant $H:=20$.
Define the function
\[
   p(s):=
   \begin{cases}
   -s+\dfrac H2s^2, & 0\le s\le 1/H,\\[2mm]
   -\dfrac1{2H}, & 1/H\le s\le 1.
   \end{cases}
\]
The following properties of $p(s)$ are easily verified. 
\begin{enumerate}
\item    $-s\le p(s)\le 0$ for all $s\in [0,1]$.
\item The derivative $p'(s) \in [-1,0]$ and is continuous for all $s\in [0,1]$. 
\item The second derivative satisfies $p''(s)=H$ on $s\in (0,1/H)$ and $p''(s)=0$ for $s>1/H$. 
\end{enumerate}
Note that the second derivative is discontinuous at $s=1/H$, but this will not affect our argument, since we only use $p''(s)$ for $s\leq 1/2H$.
Intuitively, a key property of this function is that it is highly convex in $[0,1/H]$ and yet $p(s)$ and $p'(s)$ are bounded (independent of $H$).

For a fractional point $x$, set $\Phi(x):=\sum_i p(x_i)$, and define the potential \[L(x):=F(x)+\Phi(x).
\]
As $\sum_i x_i=r$, the first property $-s\le p(s)\le0$ gives that 
\[
   -r\le \Phi(x)\le0.
\]
For $x$, let $
   T(x):=\{i\in E:0<x_i\le 1/(2H)\}$ denote the set {\em tiny} coordinates, and note that they are safely inside the quadratic region of $p$.  

We call a coordinate active if $0< x_i < 1$. By the zero-deletion step and singleton-part contraction, we can assume that at any time, the support consists only of active coordinates.

\subsection{Progress Argument}
\label{subsec:first-second-compression}
We now describe the progress argument outlined in Section \ref{sec:introduction}.

For every active coordinate $i$ of $x$, define
\[
   g_i:= \frac {\partial L(x)}{\partial x_i} = \ell_i+p'(x_i).
\]
Then, for every tangent direction $\delta$,
 the differential  $DL(x)[\delta]=\sum_i\delta_i g_i$.
We have the following useful fact. 

\begin{lemma}[First-order progress or leverage control]
\label{lem:first-order-gap-or-leverage-control}
For any nondegenerate $x$,
 there is either a unit tangent direction $\delta$ with   $DL(x)[\delta]>1/\sqrt2$,
or else $\ell_i\le 3$ for every active coordinate $i$.
\end{lemma}
\begin{proof}
Suppose that some part $P_a$ has active coordinates
$i,j\in P_a\cap E$ such that $g_i-g_j>1$.
Then take
   $\delta:=(e_i-e_j)/\sqrt{2}$.
This direction is tangent to $\Delta_r$, has $\|\delta\|_2=1$, and satisfies
\[
   DL(x)[\delta]=(g_i-g_j)/{\sqrt2}>1/{\sqrt2}.
\]
Otherwise, it must be that within every active part, the marginals differ by at most $1$.
Fix a part $P_a$.  The $x$-weighted average of
$g_i$ in this part is
\[
   \sum_{i\in P_a}x_ig_i
   =\sum_{i\in P_a}x_i\ell_i+
     \sum_{i\in P_a}x_ip'(x_i)
   =1+\sum_{i\in P_a}x_ip'(x_i)
   \le 1,
\]
because $p'\le0$.  As all active $g_i$'s in $P_a$ differ by at most $1$, every active
$i\in P_a$ satisfies
   $g_i\le 2$.
Using $p'\ge -1$, we obtain
   $\ell_i=g_i-p'(x_i)\le 2+1=3$.
\end{proof}

Call $x$ {\em approximately stationary} if there is no unit tangent direction $\delta$ with $DL(x)[\delta]>1/\sqrt2$. We will show that in this case, if the support is large, one can make large second-order progress on tiny coordinates. 

Recall the definitions, \[
   C(\delta):=\sum_i\delta_iQ_i, \,\,C_a(\delta):=\sum_{i\in P_a}\delta_iQ_i  
\text{ and } \eta_a(\delta):=\tr C_a(\delta)-\tr(B_aC(\delta)).\]
Call a direction $\delta$ {\em null-Schur} if
$\eta(\delta)=0$.
For such directions, the Hessian of $F$ becomes
   \[D^2F(x)[\delta,\delta]
   =-\tr(C(\delta)^2).\]

\begin{lemma}
\label{lem:low-rayleigh-null-schur-continuous} If $x$ is approximately stationary, and the number of tiny coordinates
   $|T(x)|>4r$,
then there is a unit tangent direction $\delta$, supported on $T(x)$, such that
\[\eta(\delta)=0  \text{ and }
   \tr(C(\delta)^2)\le 18.\]
In particular, there is a direction with $DL(X)[\delta]\geq 0$ and
\[
   D^2L(x)[\delta,\delta]
   =H-\tr(C(\delta)^2)
   \ge 2.
\]
\end{lemma}

\begin{proof}
Let $T:=T(x)$.  Consider the linear subspace
\[
   K:=
   \Big\{
   \delta\in\mathbb R^T:
   \sum_{i\in P_a\cap T}\delta_i=0\ \forall a,
   \quad
   \eta(\delta)=0
   \Big\}.
\]
The equations for parts impose at most $r$ linear constraints, and
$\eta(\delta)=0$ imposes at most $r-1$ linear constraints, as
$\eta(\delta)\in\mathbf 1^\perp$. As $|T|>4r$, this gives
\[ \dim(K)\ge |T|-2r+1 > |T|/2.\]
Now consider the positive semidefinite quadratic form
\[
   R(\delta):=\tr(C(\delta)^2) = \tr\Big((\sum_{i\in T} \delta_i Q_i)^2\Big)
   =\Big\|\sum_{i\in T}\delta_iQ_i\Big\|_F^2
\]
on $\mathbb R^T$.  Its trace in the standard coordinate basis is
\[
   \sum_{i\in T}R(e_i)
   =\sum_{i\in T}\tr(Q_i^2) = \sum_{i\in T} \ell_i^2 \leq 9|T|,
\]
where we use that $Q_i= q_iq_i^\top=M^{-1/2} w_iw_i^\top M^{-1/2}$ is rank one, and thus
   $\tr(Q_i^2)=(\tr Q_i)^2=\ell_i^2$, and that $\ell_i \leq 3$ for each active $i$ by Lemma~\ref{lem:first-order-gap-or-leverage-control}.

So the average eigenvalue of the restriction of $R$ to the subspace $K$ is at most
   $9|T|/\dim(K)<18$, and there exists some unit direction $\delta\in K$ with 
   $\tr(C(\delta)^2)\leq 18$, and hence
   \[D^2 F(x)[\delta,\delta] = -\tr(C(\delta)^2) \geq -18.\]
   As $\delta$ is supported on $T$, and $p''(x_i)=H$ for tiny $i$, we have  $D^2 \Phi(x)[\delta,\delta]=H\|\delta\|_2^2=H$ and thus
   \[D^2 L(x)[\delta,\delta] \geq H-18 = 2.\qedhere \]
\end{proof}
Thus, as long as a nondegenerate rank-$r$ component has more than $(2H+4)r$ active coordinates, either Lemma~\ref{lem:first-order-gap-or-leverage-control} gives a first-order ascent direction or Lemma~\ref{lem:low-rayleigh-null-schur-continuous} gives a second-order ascent direction. Indeed, outside $T(x)$ every active coordinate has mass larger than $1/(2H)$, so there are at most $2Hr$ such coordinates. Following these local ascent directions keeps $L$ nondecreasing. Since $-r\le \Phi\le0$, the resulting loss in the log-objective $F$ on a rank-$r$ component is at most $O(r)$; exact contractions and tight-set splits do not change the relaxation value, so the total loss over all residual components is $O(r)$.

\subsection{Terminal rounding}
\label{sec:terminal}
Once the support of $x$ has size $O(r)$, we can apply the $\exp(-O(n))$ approximation algorithm of \cite{EbrahimiStraszakVishnoi2017}, which results in only an $\exp(-O(r))$ factor loss because the residual support size is $O(r)$.
More formally, we use the following result.
\begin{theorem}[\cite{EbrahimiStraszakVishnoi2017}]
    Let \(A\in \mathbb R^{n\times n}\) be PSD, and let $\mathcal M$ be a rank-one partition matroid on \([n]\), with rank $r$ and parts $P_a$ where each part has rank $1$. Let
\[
        \mathrm{OPT}:=\max_{B\text{ base of }\mathcal M}\det(A_{B,B}).
\]
Then there is a randomized polynomial-time algorithm that outputs a base $S$ of $\mathcal M$ such that, with  probability at least $3/4$,
\[
        \det(A_{S,S})
        \ge
        \mathrm{OPT}\cdot (2e)^{-2r}
        \prod_{a=1}^r |P_a|^{-1}.
\]
\end{theorem}
In particular, if the compressed residual instance has total support $m=O(r)$, then
$\prod_{a=1}^r |P_a|\le (m/r)^r=O(1)^r$ by AM--GM, and the terminal rounding loss is only $\exp(O(r))$.

For completeness, we give a self-contained proof of this result in Appendix \ref{sec:terminal-esv}.

%% file: finite-steps-no-highlights.tex
\section{Polynomial-time Discrete Dynamics}
\label{sec:finite-dynamics}
We now turn the continuous compression argument into a polynomial-time finite-step procedure, proving Theorem \ref{thm:main}.
We use the same potential
$L(x)=F(x)+\Phi(x)$
and the same dimension-counting argument as previously. However, we need several technical ideas to ensure that each step makes sufficient progress.

A key problem is that when some $x_i$ approaches $0$, the z-Hessian $\Gamma^{-1}$ may become more unstable, and the Schur term $\eta(\delta)^\top \Gamma^{-1} \eta(\delta)$ can change substantially when we move from $x$ to $x+t\delta$, forcing the step size $t$ to be small. To get around this, we will work with the single floor $\Delta_r(\tau):=\{x \in \Delta_r:x_i \geq \tau\text{ for every active } i\}$
and delete a coordinate $i$ whenever  $x_i=\tau$. Thus the same parameter $\tau$ serves both as the computational floor and as the deletion threshold. 
To handle this, we modify the arguments in Section \ref{sec:continuous-compression}, to work with $x \in \Delta_r(\tau)$, and to make them robust to step sizes of $1/\poly(n,d)$. Crucially, we will not analyze the stability of $\Gamma^{-1}$ directly, but instead we use a probability interpretation of the derivatives of $F$, based on the Cauchy-Binet terms, which leads to substantially cleaner computations.

We now give the details.

Let $\tau:=1/(10^4dn)$ be a small inverse-polynomial deletion threshold. We use the same constant $H$ as in Section~\ref{sec:continuous-compression}.

\subsection{Initialization}
We begin by solving the relaxation \eqref{eq:prog-b}, to compute   $F^\star:=\max_{x\in\Delta_d}F(x)$ and an approximately optimal solution $x^*$ using the algorithm in \cite{MadanNikolovSinghTantipongpipat2020}.

We contract singleton parts and split all proper exact tight subsets of parts, and also delete the zero-coordinates. After this, every residual instance is nondegenerate.

Fix some residual instance of rank $r$, and let us (re)use $x^*$ to denote the optimizer for it. We first transform $x^*$ to lie in
        \[\Delta_r(\tau)
        :=
        \{x\in\Delta_r: x_i\ge\tau\text{ for every active }i\}.\]
without significantly affecting $F(x^*)$, as follows.
\begin{lemma}[Interior restriction]
\label{lem:interior-restriction-balanced}
Any $x \in \Delta_r$ can be transformed into $y \in \Delta_r(\tau)$ such that $F(y) \geq F(x) + r \log (1-n\tau)$.
\end{lemma}
\begin{proof}
Let $u$ be the point with
$u_i:=1/|P_{a(i)}|$ for each $i$. Then $u(P_a)=1$ for each part, and hence $u\in \Delta_r$. Let
$\varepsilon:=n\tau$, and consider 
        $y=(1-\varepsilon)x+\varepsilon u$.
Then $y\in \Delta_r(\tau)$, as each $y_i\geq \tau$ (as each $|P_a|\le n$).
Moreover, as
        $A_a(y)\succeq(1-\varepsilon)A_a(x)
$ for each $a\in [r]$, every product-one scaling loses at most the factor $(1-\varepsilon)^r$ in determinant. 
Taking the infimum over scalings gives the claim.
\end{proof}
As the loss $O(rn\tau)$ is negligible, we can assume without loss of generality that $x^* \in \Delta_r(\tau)$.

\subsection{A Concave Surrogate for $L$ and Progress}
The continuous proof used the stationarity of $L(x)=F(x)+\Phi(x)$. Fix a point $x$. As $F$ is concave while $\Phi$ is convex, it will be convenient to replace $L(y)$ in the neighborhood of $x$, by a concave surrogate
\begin{equation}
\label{eq:proximal-surrogate-balanced}
        \Psi_x(u)
        :=
        F(u)+\ip{\nabla\Phi(x)}{u-x}
        -\frac H2\norm{u-x}_2^2.
\end{equation}
Note that $\Psi_x(u)$ is strictly concave in $u$. It also satisfies that 
$\Psi_x(u) +\Phi(x) \leq L(u)$ for all $u$.

Fix the current active support $E$. To solve the next step, we will solve an optimization problem over the region.
\begin{equation}
\label{eq:truncated-face-balanced}
        \mathcal K_E
        :=
        \left\{u:
        u(P_a)=1\ \forall a,
        \ u_i\ge \tau\ \forall i\in E
        \right\}.
\end{equation}
For $x\in\mathcal K_E$, the next step is given by the unique optimizer 
\begin{equation}
    y = \argmax_{u\in \mathcal K_E} \Psi_x(u).
\end{equation}
We have the following analogue of Lemma \ref{lem:first-order-gap-or-leverage-control}.

\begin{lemma}[Progress or approximate stationarity]
\label{lem:proximal-progress-balanced}
At the next point $y$, either $\|y-x\|_2 \geq 1/(5H)$, in which case $L(y)-L(x)
        \ge 1/(50H)$, or otherwise, the leverage $\ell_i(y) \leq 3$ for
        every $i \in E$.
\end{lemma}
\begin{proof}
Since $x\in\mathcal K_E$, by optimality of $y$ we have $\Psi_x(y)\geq \Psi_x(x)$, which gives that
\[  F(y)-F(x)+\ip{\nabla\Phi(x)}{y-x}
        \ge
        \frac H2\norm{y-x}_2^2.\]
As $\Phi$ is convex, $\Phi(y)-\Phi(x)\ge\ip{\nabla\Phi(x)}{y-x}$,
which gives that \[L(y)-L(x) \geq \frac{H}{2} \|y-x\|_2^2.\]
If $\|y-x\|_2 \geq 1/(5H)$, this already gives $L(y)-L(x)\geq 1/(50H)$ and we are done. So, let us assume that $\|y-x\|_2 < 1/(5H)$. We use this to bound the $\ell_i$.

Consider the surrogate gradient coordinate at $y$,
\begin{equation}    
\label{eq:deriative-hi}
h_i:=\partial_i\Psi_x(y)
        =\ell_i(y)+p'(x_i)-H(y_i-x_i)
        \end{equation}
A key observation is the following. Fix some part $a$. Then, there is a $\beta_a$ such that for all $j\in P_a$ with $y_j > \tau$, we have $h_j=\beta_a$. For all variables $i\in P_a$ such that $y_i=\tau$, we have  that $h_i \leq \beta_a$. To see this, consider the Lagrangian of the optimization problem for $y$:
\[ \Psi_x(y) + \sum_a \beta_a (1-\sum_{i\in P_a} y_i) + \sum_i \lambda_i (y_i - \tau),\]
where $\beta_a \in \R$ and $\lambda_i \geq 0$. Taking the partial derivative with respect to $y_i$ gives $h_i  = \beta_{a(i)} - \lambda_i \leq \beta_a$.

Next, observe that by Lemma \ref{lem:partwise-balance} about partwise-balance, for each part $a$ we have
\[ \sum_{i\in P_a} y_i \ell_i(y) =1.\]
But as $y(P_a)=1$, the coordinates with $y_i = \tau$ can have a total mass of at most $n\tau < 10^{-4}$. Thus, by averaging, there is some coordinate $j\in P_a$ with $y_j> \tau$ and $\ell_j(y) \leq (1-10^{-4})^{-1}$. Fix this $j$. By the observation above, for each $i \in P_a$, we have $h_i \leq h_j$. Hence, plugging the expression for $h_i, h_j$ using \eqref{eq:deriative-hi} gives
\[ \ell_i(y) +p'(x_i) - H(y_i-x_i) \leq \ell_j(y) +p'(x_j) - H(y_j-x_j).\]
Noting that
the gradient of the true potential $L(y)$ is $\partial_iL(y)=\ell_i(y)+p'(y_i)$, plugging this and rearranging, the above gives
\begin{align*}
\partial_iL(y) - \partial_jL(y) &\leq 
p'(y_i)- p'(x_i) - (p'(y_j) - p'(x_j)) +H(y_i-x_i - (y_j-x_j))  \\
& = \langle \nabla \Phi(y)-\nabla \Phi(x) +H(y-x), e_i-e_j \rangle \\
& \leq 2\sqrt{2} H \|y-x\|_2 \leq 2\sqrt{2}/5.
\end{align*}
Here the second step uses that $p'(y_i) = (\nabla \Phi(y))_i$ and the third step uses that $p'$ is $H$-Lipschitz and that $|e_i-e_j|\leq \sqrt{2}$ and our assumption $\|y-x\|_2 < 1/(5H)$.

As $p'(\cdot)\in [-1,0]$, this gives that $\ell_i(y)-\ell_j(y) \leq 2\sqrt{2}/5+1$. Together with the bound on $\ell_j(y)$, this gives $\ell_i(y) \leq (1-10^{-4})^{-1} + 1+ 2\sqrt{2}/5 \leq 3$, as desired.
\end{proof}

The Lemma gives a win-win argument: a large move gives the fixed gain $L(y)-L(x)\ge 1/(50H)= \Omega(1)$, which can only happen $O(r)$ times during the course of the algorithm, since $\Phi$ never exceeds $r$. On the other hand, a small move gives a leverage bound, which will be exploited crucially.

\subsection{Deleting small coordinates}
Once a coordinate reaches $\tau$, we can use this leverage bound to safely delete it (if a coordinate reaches $\tau$ and a move is large, we wait until the next small move). Without loss of generality, we can assume that we delete one coordinate at a time. Suppose $i\in P_a$ and $x_i=\tau$.

By Lemma \ref{lem:proximal-progress-balanced}, we know that the leverage $\ell_i \le 3$.  Let $\overline x$ be obtained by deleting $i$ and renormalizing its part:
\begin{equation}
\label{eq:one-delete-renormalize-balanced}
        \overline x_i=0,
        \qquad
        \overline x_k=x_k/(1-\tau)\quad(k\in P_a\setminus\{i\}),
        \qquad
        \overline x_k=x_k\quad(k\notin P_a).
\end{equation}

The following lemma shows that deleting such a coordinate $i$ affects the objective $F(x)$ and the potential $L(x)$ negligibly.
\begin{lemma}[One-coordinate deletion]
\label{lem:finite-small-deletion-balanced}
Suppose $x$ is nondegenerate and $\ell_i\le 3$ for the coordinate $i$ being deleted. Then \[F(\overline x)\ge F(x)+r\log(1-2\tau\ell_i) \text{ and } 
        L(\overline x)
        \ge
        L(x)-(12r+1)\tau.\]
        In particular, as $F(\overline{x}) > -\infty$, the new support is still admissible.
\end{lemma}
\begin{proof}
Using notation from Section \ref{sec:derivatives}, let us denote $Q_i=M^{-1/2}w_iw_i^\top M^{-1/2}$, where $w_i = \exp(z(x)_{a(i)}/2) v_i$ and $B_a=\sum_{i\in P_a}x_iQ_i$ at the point $x$. Recall that $\sum_a B_a=I$ and by Lemma \ref{lem:partwise-balance}, we have $\tr(B_a)=1$. 

At the current scaling (depending on $z(x)$), the deletion removes
        $R:=\tau Q_i$ from $B_a$. 
        Let us denote
        \[\rho:=\tr R=\tau \tr Q_i=\tau \ell_i.\]
Since $Q_i=q_iq_i^\top$ is rank one, $R$ has exactly one nonzero eigenvalue $\rho$.  Thus, if we kept the old scaling the determinant $\det(M)$ would lose the factor $
        \det(I-R)=1-\rho$.
However, we need to show that the inner minimization over product-one scalings at $\overline x$ cannot make the determinant much smaller.

\medskip
Let us define $D_a:=B_a-R$ and $D_b:=B_b$ for $b\ne a$.
Then $\sum_b D_b=I-R$. Define the normalized version, $
        \widetilde D_b:=(I-R)^{-1/2}D_b (I-R)^{-1/2}$ so that 
        $\sum_b\widetilde D_b=I$, and 
let $\alpha_b:=\tr\widetilde D_b$. 
Then,
\[
        \alpha_b-1
        =\tr(((I-R)^{-1}-I)B_b) \text{ for $b\ne a$,  and } 
  \alpha_a-1
        =\tr(((I-R)^{-1}-I)B_a)-\tr((I-R)^{-1}R).\]
As $\rho<1/4$, a simple computation gives that $\|\alpha-\one\|_1
        \le 2\rho/(1-\rho)$.

We use the following matrix bound:
Let $E_1,\ldots,E_r\succeq0$ satisfy
$\sum_{b=1}^r E_b=I$. Let $\alpha \in \R^r$ be the vector with coordinates
        $\alpha_b:=\tr E_b$, and let  $\epsilon:=\|\alpha-\one\|_1$.
Then, for every $q_1,\ldots,q_r>0$,
\begin{equation}
\label{eq:elementary-matrix-product-bound-expanded}
        \det\Big(\sum_b q_bE_b\Big)
        \ge
        \left(1-\epsilon/2\right)^r\prod_b q_b .
\end{equation}
Now fix any positive product-one scaling $q$, so $\prod_b q_b=1$.  As $D_b=
        (I-R)^{1/2}\widetilde D_b (I-R)^{1/2}$ for each $b$, and
applying \eqref{eq:elementary-matrix-product-bound-expanded} to $E_b=\widetilde D_b$ gives,
\[
        \det\Big(\sum_bq_bD_b\Big)
        =
        \det(I-R)\cdot
        \det\Big(\sum_bq_b\widetilde D_b\Big) \geq 
        (1-\rho)
        \Big(1-\frac{\rho}{1-\rho}\Big)^r \geq 
        (1-2\rho)^r.\]
As this holds for every product-one $q$, the unnormalized deletion decreases $F$ by at most $-r\log(1-2\rho) = -r \log (1-2\tau \ell_i)$.
Finally, the renormalization \eqref{eq:one-delete-renormalize-balanced} rescales part $a$ by $1/(1-\tau)>1$, which only increases $F$.  This proves the lower bound bound on $F(\overline{x})$.

We now lower bound $L(\overline{x}) = F(\overline{x})+\Phi(\overline{x})$.
 Using that $\log(1-2 \tau \ell_i)\ge-4\tau \ell_i$ as $0\le \tau \ell_i\le1/4$, and that $\ell_i\le 3$, the decrease in $F(\overline{x})$ is at most $12 r\tau$.
For $\Phi$, setting $x_i$ to zero removes $p(s)\le0$ and therefore does not decrease $\Phi$.  Renormalizing the surviving coordinates in part $a$ adds total mass exactly $\tau$, and as $|p'|\le1$ the $\Phi$ can decrease by at most $\tau$.  
\end{proof}

As there are at most $n$ deletions, the total loss over the course of the algorithm is $O(nr\tau) \ll 1$ which is negligible.

\subsection{Probabilistic interpretation of derivatives}
\label{subsec:probabilistic-derivative-dictionary}
 We now describe a probabilistic interpretation of the derivatives in Section \ref{sec:derivatives}, based on the Cauchy-Binet formula.
This dictionary will be crucial for proving the finite-step version of the null-Schur argument.

Recall that
 $M(x,z) = \sum_i x_i e^{z_{a(i)}} v_iv_i^\top$ and $\psi(x,z) = \log \det M(x,z)$. 
Fix a non-degenerate $x$ and let $z=z(x)$. We write $M$ for $M(x,z(x))$ henceforth.

 By the Cauchy-Binet formula, we have
 \begin{equation}
 \label{eq:prob-detm}
 \det M  = \sum_{T: |T|=r} \det(V_T)^2 \prod_{i\in T}x_i e^{z_{a(i)}} = \sum_{T: |T|=r} \det(V_T)^2 e^{\langle R(T),z \rangle} \prod_{i\in T}x_i 
 \end{equation}
        where $V_T$ is the $r \times r $ matrix with columns $v_i$ for $i\in T$, and for a set $T$ we define its part-count vector
\[R(T) \in \R^r, \text{ 
    with }R_a(T):=|T\cap P_a|.\]
Note that a term is positive only if the vectors in $T$ are independent. 
This defines the Cauchy--Binet distribution on such sets,
\begin{equation}
\label{eq:mux}
           \mu_x(T)
        :=
        \frac{\det(V_T)^2 e^{\langle R(T),z \rangle}\prod_{i\in T}x_i}{\det M}.
\end{equation}
Next, for a tangent $x$-direction
$\delta$, define
      \begin{equation}
      \gamma_i:=\frac{\delta_i}{x_i}, \qquad \sigma_\delta(T):=\sum_{i\in T}\gamma_i,
        \qquad
        a_\delta(T):=\sum_{i\in T}\gamma_i^2.
        \label{eq:gst}
        \end{equation}
Recall the terms $C(\delta),C_a(\delta),B_a, \ell_i$ and $\eta_a(\delta)$ from Section \ref{sec:derivatives}.
\begin{lemma}
\label{lem:prob-derivatives}
The derivatives of $\psi(x,z)$ in Lemma \ref{lem:derivatives-psi} can be expressed in a probabilistic way based on the Cauchy--Binet distribution $\mu_x$ as follows:
\begin{enumerate}
    \item $\frac{\partial \psi(x,z)}{\partial z_a} =\Ex R_a $. 
     \item    $D_x \psi(x,z)[\delta] = \Ex \sigma_\delta$.
    \item $D_x \left(\frac{\partial \psi}{\partial z_a}\right)[\delta] = \operatorname{Cov}(R_a,\sigma_\delta)$.
    \item $D^2_{xx}\psi(x,z)[\delta,\delta] = \operatorname{Var}(\sigma_\delta)-\Ex a_\delta$.
\end{enumerate} 
In particular this gives
\begin{align}
\label{eq:prob-derivatives}
\Ex R_a= \tr(B_a),    \quad \Ex \sigma_\delta = \tr(C(\delta)), \quad  \operatorname{Cov}(R_a,\sigma_\delta) = \eta_a(\delta), 
    \quad \operatorname{Var}(\sigma_\delta)-\Ex a_\delta=-\tr(C(\delta)^2).
\end{align}
Moreover, by Lemma \ref{lem:partwise-balance}, at the point $(x,z(x))$, we have that $\Ex R_a = \tr(B_a)=1$.
\end{lemma}
\begin{proof}
Each of these identities follows directly  using the Cauchy--Binet expansion \eqref{eq:prob-detm} and the definition of the distribution $\mu_x$ in \eqref{eq:mux}. In particular,
\[\frac{\partial \log \det M}{\partial z_a}  =\sum_T \mu_x(T) R_a(T)  = \Ex R_a, \text{ and } 
 \frac{\partial \log \det M}{\partial x_i} = \sum_{T\ni i} \frac{\mu_x(T)}{x_i},\] and using the definition of $\sigma_\delta$, this gives  $D_x \psi(x,z)[\delta] = \Ex \sigma_\delta$.

 Differentiating the first expression with respect to $x$ again, in the direction $\delta$, differentiates the normalized weights and gives 
        $\operatorname{Cov}(R_a,\sigma_\delta)$.
Similarly differentiating the  second term again with respect $x_j$ and considering the direction $(\delta,\delta)$ gives $\operatorname{Var}(\sigma_\delta)-\Ex a_\delta$.
\end{proof}

\subsection{Bounding $F$ in null-Schur direction via Gibbs lower bound}
We now prove the finite-step equivalent of Lemma \ref{lem:low-rayleigh-null-schur-continuous},  which lower bounds the change $F(x)$ when $x$ moves to $x+t\delta$, where $\delta$ is a null-Schur direction satisfying $\eta(\delta)=0$. Note crucially that $\eta$ is computed at $x$. 

The proof will crucially exploit the multilinearity of the determinant and the probabilistic interpretation of the derivatives in Section \ref{subsec:probabilistic-derivative-dictionary}. This will allow us to avoid directly dealing with stability of the $z$-Hessian in null-Schur directions.

\begin{lemma}[Finite null-Schur step]
\label{lem:finite-null-schur-estimate-balanced}
Let $x$ be nondegenerate. Let $\delta$ be any feasible direction with $\delta(P_a)=0$ for each part $a$, such that $\eta(\delta)=0$.
Let  $K_0:=\max_{i\in\supp(x)} |\delta_i|/x_i$. 
Then 
for any $t$ with
        $|t|\le 1/(4rK_0)$, \[
        F(x+t\delta)
        \ge{}
        F(x)+tDF(x)[\delta]
        -\frac{t^2}{2}\tr(C(\delta)^2)
        -cr^3K_0^3|t|^3,\]
where $c$ is a universal constant.
\end{lemma}
Before proving this, let us recall two basic facts.

Let $\mu, \nu$ be distributions on a discrete set $\Omega$, such that $\nu$ is absolutely continuous with respect to $\mu$, i.e., $\nu(y)=0$ whenever $\mu(y)=0$. The relative entropy or KL divergence of $\nu$ with respect to $\mu$ is defined as
\[  D(\nu ||\mu) = \sum_y  \nu(y) \log\ \frac{\nu(y)}{\mu(y)} = \E_\nu \log \frac{d\nu}{d\mu}\]
One has the following Donsker--Varadhan/Gibbs variational inequality.
For any function $f: \Omega \rightarrow \R$,
\[ \log \Ex_\mu  [e^{f}] \geq \E_\nu [f] -  D(\nu||\mu)\]
\begin{lemma}
\label{lem:balanced-certificate}
Let $\mu,\nu$ be probability distributions on a finite set $\Omega$. Let $R:\Omega\to\R^r$, and $u:\Omega\to\R$.  Suppose $\nu$ is absolutely continuous with respect to $\mu$, and $\Ex_\nu [R]=\one.$
Then
\begin{equation}
\label{eq:balanced-certificate}
        \inf_{s\in\one^\perp}
        \log\Ex_{T\sim\mu}
        \exp\{u(T)+\ip{s}{R(T)-\one}\}
        \ge
        \Ex_\nu [u]-D(\nu\|\mu).
\end{equation}
\end{lemma}

\begin{proof}
For any fixed $s\in\one^\perp$, Gibbs variational lower bound gives
\[
        \log\Ex_\mu e^{u+\ip{s}{R-\one}}
        \ge
        \Ex_\nu\bigl[u+\ip{s}{R-\one}\bigr]-D(\nu\|\mu).\]
Because $\Ex_\nu R=\one$, the term
        $\Ex_\nu\ip{s}{R-\one}=0$.
So the resulting lower bound is independent of $s$, and remains true after taking the infimum over $s\in\one^\perp$.
\end{proof}

 We now prove Lemma~\ref{lem:finite-null-schur-estimate-balanced}. We use the notation from Subsection~\ref{subsec:probabilistic-derivative-dictionary}.  Thus $T\sim\mu_x$, $R(T)$ is the part-count vector, $\gamma_i=\delta_i/x_i$, $\sigma_\delta(T)=\sum_{i\in T}\gamma_i$, and $a_\delta(T)=\sum_{i\in T}\gamma_i^2$.
\begin{proof}
At the old balanced point $x$, $F(x)=\log\det M$. Consider the perturbed point $x+t\delta$, and consider an arbitrary $\widetilde{z} = z(x)+s$ with  $s\in\one^\perp$. Let us denote $\widetilde{M} := M(x+t\delta,z(x)+s)$.

By Cauchy--Binet and the definition of $\mu_x$, we have
\begin{equation}
    \label{eq:m-ratio}
        \frac{\det(\widetilde{M})}{\det(M)} = \frac{
        \det\left(\sum_a e^{s_a} e^{z_a}\sum_{i\in P_a}(x_i+t\delta_i)v_iv_i^\top\right)
        }{\det M} =
        \Ex_{\mu_x}
        \Big[
        \exp\bigl(\ip{s}{R(T)}\bigr)
        \prod_{i\in T}(1+t\gamma_i)
        \Big].
\end{equation}
We denote $\mu_x$ by $\mu$ to avoid subscript overload below.

Set \begin{equation} u(T) := u_t(T)=\sum_{i\in T}\log(1+t\gamma_i).
\label{eq:yoo-t}
\end{equation}
Since $s\in\one^\perp$, we have $\ip{s}{R(T)}=\ip{s}{R(T)-\one}$.  Therefore, taking logarithm in \eqref{eq:m-ratio}, and minimizing over $s$ to obtain the inner minimum defining $z(x+t\delta)$, we have
\begin{equation}
\label{eq:finite-mgf-identity-balanced}
        F(x+t\delta)-F(x)
        =
        \inf_{s\in\one^\perp}
        \log\Ex_{\mu} \Big[ \exp\left(
        u(T)+\ip{s}{R(T)-\one}
        \right) \Big].
\end{equation}
This gives an exact formula for the decrease, but we need to lower bound this decrease over all possible $s \in \one^\perp$. We will do this by considering a suitable tilt of the distribution of $\mu$.

For a set $T$, let us denote $X(T):=\sigma_\delta(T)-\Ex_{T\sim \mu} [\sigma_\delta(T)]$. 
Define the tilted distribution $\nu$ as
\begin{equation}
\label{eq:nut}
        \frac{d\nu}{d\mu}(T):=1+tX(T).
        \end{equation}
We first verify that $\nu$ is a valid distribution.
For any $T$, note that $|\sigma_\delta(T)| = |\sum_{i\in T} \gamma_i| \le rK_0$, giving the absolute bound $|X(T)|\le2rK_0$. As $|t|\le1/(4rK_0)$, this gives that the density $1+tX$ lies between $1/2$ and $3/2$. 
Additionally as $\E_{T\sim \mu} X(T)=0$, it follows that $\nu$ is a distribution and is absolutely continuous with respect to $\mu$.

Next, $\nu$ satisfies the condition $\Ex_{\nu} R=\one$ needed in Lemma \ref{lem:balanced-certificate}. To see this,
\begin{align*}
        \Ex_{\nu}(R-\one)
        &=
        \Ex_{\mu}\bigl[(1+tX)(R-\one)\bigr]
        = \Ex_{\mu}\bigl[(tX)(R-\one)\bigr] \\
        &= \Ex_{\mu}\bigl[tXR\bigr]
        = t\operatorname{Cov}_{\mu}(R,\sigma_\delta)
        = t\eta(\delta)
        =0.
\end{align*}
Here the first equality uses \eqref{eq:nut} and the second equality uses that $\E_{\mu} R = \one$ by Lemma \ref{lem:prob-derivatives} and \ref{lem:partwise-balance}. The third equality uses that $\Ex_{\mu} X= \E_{\mu}[\sigma_\delta] - \E_{\mu}[\sigma_\delta] =0$ and the second last equality uses the mixed-derivative and covariance connection in Lemma \ref{lem:prob-derivatives}.

Applying Lemma~\ref{lem:balanced-certificate} and using \eqref{eq:finite-mgf-identity-balanced} gives the lower bound
\begin{equation}
\label{eq:certificate-applied-balanced}
        F(x+t\delta)-F(x)
        \ge
        \Ex_{\nu}[u]-D(\nu\|\mu).
\end{equation}
We now bound the terms on the right suitably. 

Since $|t\gamma_i|\le1/4$ for every $i$, by Taylor expansion $\log (1+t\gamma_i) = t\gamma_i - t^2 \gamma_i^2/2 + O(|t|^3 |\gamma_i|^3)$.  Using \eqref{eq:yoo-t}, and the definition of $\sigma_\delta(T)$ and $a_\delta(T)$ in \eqref{eq:gst} and that $|\gamma_i|\leq K_0$, this gives
\begin{align}
        u(T)
        = \sum_{i\in T}\log(1+t\gamma_i) 
        = t\sigma_\delta(T)-\frac{t^2}{2}a_\delta(T)+O(rK_0^3|t|^3), \label{eq:yoo-t-taylor}
        \end{align} and thus  $\E_{\mu} [u] = t \E_{\mu}[\sigma_\delta] - (t^2/2) \E_{\mu} a_\delta+O(rK_0^3|t|^3)$. 
        
Using $d\nu/d\mu=1+tX$, this gives 
\begin{align}
        \Ex_{\nu}[u]
        &=\Ex_{\mu} [u]  + \Ex_{\mu} [t u X] \nonumber \\
        & = t \E_{\mu}[\sigma_\delta] -\frac{t^2}{2}\Ex a_\delta+ t^2\operatorname{Var}(\sigma_\delta)
        +O(r^3K_0^3|t|^3). \label{eq:nu-u-mean}
\end{align}
where second step used that $\E_{\mu}[X]=0$ and that $\Ex_{\mu} [X\sigma_\delta]= \operatorname{Var}(\sigma_\delta)$; the remaining terms are bounded using $a_\delta\le rK_0^2$ and $|X|\le2rK_0$.

To bound the relative entropy term, we use that  $(1+y)\log(1+y)=y+y^2/2+O(|y|^3)$ for $|y|\le1/2$. Setting $y=tX$ and using that  $\E_{\mu} X=0$, this  gives,
\begin{equation}
       D(\nu\|\mu)
        =
        \frac{t^2}{2}\operatorname{Var}(\sigma_\delta)+O(r^3K_0^3|t|^3).
\label{eq:rel-ent-nu-u}
\end{equation}
Plugging the estimates \eqref{eq:nu-u-mean}, \eqref{eq:rel-ent-nu-u} in \eqref{eq:certificate-applied-balanced} gives
\[
        F(x+t\delta)-F(x)
        \ge
        t\Ex_\mu[\sigma_\delta]+\frac{t^2}{2}\bigl(\operatorname{Var}(\sigma_\delta)-\Ex a_\delta\bigr)
        -O(r^3K_0^3|t|^3).
\]
Using that  $\Ex_\mu[\sigma_\delta]=DF(x)[\delta]$

and
$
        \operatorname{Var}(\sigma_\delta)-\Ex a_\delta=-\tr(C(\delta)^2)$
by Lemma \ref{lem:prob-derivatives} gives the result.
\end{proof}

\subsection{Second Order Progress}

We now show that, after a small rebalancing step has certified bounded leverage and after all floor coordinates have been deleted, either the current component already has small support, or a finite null-Schur step gives inverse-polynomial progress.

Define the small-coordinate set
\begin{equation}
\label{eq:small-set-balanced}
        I_{\rm sm}(x):=
        \left\{i:\tau<x_i\le 1/(2H)\right\}.
\end{equation}
If every active coordinate is larger than $\tau$, then every active coordinate outside $I_{\rm sm}(x)$ has mass larger than $1/(2H)$. Hence there are at most $2Hr$ such coordinates.

\begin{lemma}
\label{lem:small-progress-balanced}
Assume $x_i>\tau$ for every active coordinate and $\ell_i\le 3$ for every active coordinate. If
$
        |I_{\rm sm}(x)|>4r,
$
then one can compute a feasible point $x^+\in \Delta_r(\tau)$ such that either some coordinate of $x^+$ is equal to $\tau$, or
\[
        L(x^+)-L(x)\ge c_2\frac{\tau^6}{r^6}
        \ge c_2\frac{\tau^6}{d^6},
\]
for a universal constant $c_2>0$.
\end{lemma}

\begin{proof}
Let $I:=I_{\rm sm}(x)$. We argue as in the dimension-counting argument in the continuous proof. On the coordinate space supported on $I$, the equations
\[
        \delta(P_a)=0\quad(a\in[r]),
        \qquad
        \eta(\delta)=0
\]
impose at most $r+(r-1)$ independent linear constraints, because $\eta(\delta)\in\one^\perp$. Since $|I|>4r$, their common kernel has dimension larger than $|I|/2$.

The trace of the quadratic form
$       \delta\mapsto \tr(C(\delta)^2)
$
on $\R^I$ is
\[
        \sum_{i\in I}\tr(Q_i^2)=\sum_{i\in I}\ell_i^2\le 9|I|,
\]
where we use the leverage bound and the fact that each $Q_i$ is rank one. Hence the restriction of this quadratic form to the kernel has average eigenvalue less than $18$. Therefore there is a unit vector $\delta$ in the kernel such that
\[
        \eta(\delta)=0,
        \qquad
        \tr(C(\delta)^2)\le 18.
\]
Choose the sign of $\delta$ so that $DL(x)[\delta]\ge0$.

Since every active coordinate is larger than $\tau$ and $\|\delta\|_2=1$, we have
        $K_0:=\max_i|\delta_i|/x_i\le 1/\tau$.
Let $t$ be the largest value satisfying
$
        0<t\le c_1\tau^3/r^3
$
and such that $x+t\delta\in\Delta_r(\tau)$. Here $c_1>0$ is a sufficiently small universal constant. This choice guarantees $t\le 1/(4rK_0)$, so Lemma~\ref{lem:finite-null-schur-estimate-balanced} applies. It also keeps every moved coordinate inside the quadratic region of $p$, unless the step stops because a coordinate hits the floor.

Set $x^+:=x+t\delta$. If the maximality of $t$ is caused by the floor constraint, then some coordinate of $x^+$ is equal to $\tau$, and we are done. Otherwise $t=c_1\tau^3/r^3$. Since $\delta$ is supported on $I_{\rm sm}(x)$ and no coordinate leaves the quadratic region of $p$, we have
\[
        \Phi(x+t\delta)=\Phi(x)+tD\Phi(x)[\delta]+\frac{H}{2}t^2.
\]
Combining this identity with Lemma~\ref{lem:finite-null-schur-estimate-balanced} gives
\[
\begin{aligned}
        L(x+t\delta)-L(x)
        &\ge
        tDL(x)[\delta]
        +\frac{t^2}{2}\bigl(H-\tr(C(\delta)^2)\bigr)
        -cr^3K_0^3t^3 .
\end{aligned}
\]
The first term is nonnegative by the choice of sign, and $H-\tr(C(\delta)^2)\ge H-18$. Taking $c_1$ small enough makes the cubic term at most one half of the positive quadratic margin. Thus
\[
        L(x+t\delta)-L(x)\ge c_2 t^2
        = c_2\frac{\tau^6}{r^6},
\]
after decreasing the universal constant $c_2$ if necessary.
\end{proof}

\subsection{The Final Algorithm}

We now assemble the pieces. The algorithm maintains a forest of residual instances. Singleton contractions contribute an accumulated log-value $\Lambda$, while tight-set splits decompose the current relaxation value exactly. On each nondegenerate instance, the algorithm repeats the following steps.

First,  from the current point, compute the point $y$ from \eqref{eq:proximal-surrogate-balanced}. If the step is large, Lemma~\ref{lem:proximal-progress-balanced} gives a fixed increase in $L$. If the step is small, replace $x$ by $y$ and use the leverage bound from Lemma~\ref{lem:proximal-progress-balanced}. Then delete any coordinate at the floor $\tau$, using Lemma~\ref{lem:finite-small-deletion-balanced}, and immediately perform all exact singleton contractions and tight-set splits. If no deletion or structural operation is possible and $|I_{\rm sm}(x)|>4r$, apply Lemma~\ref{lem:small-progress-balanced}. Finally, if $|I_{\rm sm}(x)|\le4r$, declare the component compressed: it has at most $4r$ small coordinates and at most $2Hr$ other active coordinates, so its support size is at most $(2H+4)r$.

 Initially, the interior restriction loses only $O(dn\tau)$ in $F$, and $\sum_c\Phi(x^c)\ge -d$, so $\mathcal L\ge F^\star-O(d)$. Also, throughout the algorithm, $\mathcal L\le F^\star+O(1)$: exact contractions and tight-set splits preserve the relaxation value.
Each large rebalancing step increases $\mathcal L$ by the fixed constant $1/(50H)$. Each finite null-Schur step that does not hit the floor increases $\mathcal L$ by at least $c_2\tau^6/d^6=1/\poly(n,d)$. Each deletion removes one ground-set element and decreases $\mathcal L$ by at most $O(r\tau)$, so all deletions together lose only $O(nd\tau)=O(1)$. Singleton contractions and tight-set splits are exact. Therefore the number of large rebalancing steps and null-Schur steps is polynomial, and there are at most $n$ deletions. The process terminates after polynomially many arithmetic iterations.

Combining this with the terminal rounding theorem from Section~\ref{sec:continuous-compression} gives a transversal in each residual component with an additional $\exp(O(d))$ total multiplicative loss. The exact split and contraction identities multiply the residual determinants and add the accumulated log-value $\Lambda$. Hence the final transversal for the original instance has determinant at least $\exp(-O(d))$ times the optimum relaxation value, and therefore at least $\exp(-O(d))$ times the optimum integral value.

%% file: local-search-barrrier.tex
\section{A  barrier for cycle exchanges}
\label{sec:local-barrier}
We describe the example from~\cite{BrownLaddhaPittuSinghTetali2022},  
which shows that a local-improvement method cannot give an approximation better than $d^{\Theta(d)}$ even for partition matroids version.
In the lower bound below, we allow more general edges corresponding to $\ell$-wise swaps (\cite{BrownLaddhaPittuSinghTetali2022} considered $\ell=1$).

Fix
$\ell\ge1$ and set $\gamma:=1/\sqrt{2\ell}$.
Consider a $d\times d$ Hadamard matrix $H=[h_1|\cdots|h_d]$, and consider a subdeterminant maximization instance with $d$ parts
\[
        P_a=\{e_a,\gamma h_a\},\qquad a=1,\ldots,d,
\]
where $e_1,\ldots,e_d$ are the standard basis vectors and each part has rank $1$.  

\begin{proposition}
\label{prop:local-barrier}
The standard
transversal $E:=\{e_1,\ldots,e_d\}$
is a strict local optimum with respect to every exchange involving at most
$\ell$ parts.  However the all-Hadamard transversal $H:=\{\gamma h_1,\ldots,\gamma h_d\}$
has determinant larger by the factor
$(d/2\ell)^d.$
\end{proposition}

\begin{proof}
The standard transversal has volume $1$.  Since the Hadamard columns are
orthogonal and have norm $\sqrt d$, the all-Hadamard transversal has volume $
        \gamma^d |\det H|=\gamma^d d^{d/2}$,
and hence squared determinant $(d/(2\ell))^d$.

Suppose we exchange the parts in set $J\subseteq[d]$ with
$|J|=t\le\ell$, replacing $e_j$ by $\gamma h_j$ for $j\in J$ and keeping
$e_j$ for $j\notin J$.  Expanding along the unchanged standard-basis columns,
the volume is
\[
        \gamma^t\,|\det H[J,J]| = \gamma^t t^{t/2} \le (t/2\ell)^{t/2}<1.
\]
where the first inequality follows from Hadamard's inequality, as each column of $H[J,J]$ has length $t^{1/2}$. 
So every exchange of at most $\ell$ parts strictly decreases the volume, so $E$ is
a strict $\ell$-local optimum, even though $H$ is better
by factor at least $(d/(2\ell))^d$.
\end{proof}

%% file: terminal-appendix.tex
\section{Proof of Terminal Rounding}
\label{sec:terminal-esv}
Let $P_1,\ldots,P_r$ be the parts, each with rank $1$. We assume that vectors have dimension $r$, and index the vectors as $v_{a,i}$ where $i\in [|P_a|]$ for each $P_a$.
Let $n_a:=|P_a|$ and $n:=\sum_a n_a$.
Fix some part $P_a$. For a vector $\zeta_a \in \R^{n_a}$, define the column
\[
        C_a(\zeta_a):=\sum_{i\in P_a}\zeta_{a,i}v_{a,i}.
\]
Consider the $r\times r$ matrix $C(\zeta)$ with columns $C_a(\zeta_a)$, and let $D(\zeta) = \det C(\zeta)$.
By multilinearity of the determinant, $D(\zeta)$ is linear in each block variable $\zeta_a$, when all
other blocks are fixed. 

Call $\zeta_a$ a vertex if $\zeta_a=e_i$ for some $i\in [n_a]$.
Setting $\zeta_a=e_i$
means choosing vector $v_{a,i}$ from part $a$. 
Call $\zeta$ a vertex if each $\zeta_a$ is a vertex for $a\in [r]$.  
Define
\[
        D_{\max}:=\max_{\text{vertex }\zeta}|D(\zeta)|.
\]

For every
part $P_a$ and every $(a,i)\in P_a$, sample independently
$        U_{a,i}\sim\operatorname{Unif}[-1,1]$.
This point is not meant to be feasible.

The proof will follow from two facts.
\begin{lemma}
\label{lem:block-exposure}
$|D(U)|\ge e^{-4r}D_{\max}, $ with probability at least $3/4$ over $U$.
\end{lemma}
\begin{lemma}
\label{lem:second-terminal}
Any $U$ can be rounded efficiently to a vertex $\zeta$ such that $|D(\zeta)|\ge |D(U)|\prod_{a=1}^r n_a^{-1}$.
\end{lemma}

We need the following anti-concentration fact.
\begin{lemma}
\label{lem:one-block-uniform}
Let $U=(u_1,\ldots,u_m)$ be a random vector, with the $u_i$ uniform in $[-1,1]$ and independent.  For
any nonzero $b\in\R^m$, we have that
\[
        \Ex\left[(\log(\|b\|_\infty/|b\cdot U|))_+\right]\le1.
\]
\end{lemma}

\begin{proof}
We first show that for every
$0<\eta\le1$,
\begin{equation}
\label{eq:anti-conc-1}
      \Prb\{|b\cdot U|\le \eta \|b\|_\infty \}\le \eta.
\end{equation}
Let $j$ be such that $|b_j|=\|b\|_\infty$.  Condition on all the $U_i$ except $U_j$.  Then
        $b\cdot U=b_jU_j+C$
for a fixed $C$, and the event $|b_jU_j+C|\le\eta|b_j|$ forces $U_j$ to
lie in an interval of length at most $2\eta$ in $[-1,1]$. 
Thus the conditional probability is at most $\eta$, and the same
bound holds unconditionally.  

Setting $\eta=e^{-s}$, \eqref{eq:anti-conc-1} can be written as 
        $\Prb\left\{(\log (\|b\|_\infty)/|b\cdot U|)_+\ge s\right\}
        \le e^{-s}$, which upon
integrating gives the claimed expectation bound.
\end{proof}


\begin{proof}{(Lemma \ref{lem:block-exposure}).}
For $t=0,1,\ldots,r$, consider the exposure process
\[
        G_t(U_1,\ldots,U_t)
        :=\max_{\text{future vertex blocks }\zeta_{t+1},\ldots,\zeta_r}
        |D(U_1,\ldots,U_t,\zeta_{t+1},\ldots,\zeta_r)|.
\]
Then $G_0=D_{\max}$ and $G_r=|D(U)|$.

Fix $t$ and condition on the already exposed prefix $U_1,\ldots,U_{t-1}$.  Choose vertices in blocks
$t,t+1,\ldots,r$ attaining $G_{t-1}$.  Keep the future vertices
$\zeta_{t+1},\ldots,\zeta_{r}$ fixed, but let block $t$ vary as $x$. So
\[
       L_t(x)= D(U_1,\ldots,U_{t-1},x,\zeta_{t+1},\ldots,\zeta_r)
        =\sum_{i\in E_t}b_ix_i.
\]
for some fixed $b_i$.
Note that $L_t(e_i)=b_i$ $x=e_i$, for all $i \in [n_t]$. Let $j$ be such that
$|b_j|= \|b\|_\infty = G_{t-1}$. When $x$ is replaced by the random block $U_t$, the value $G_t$ is
at least $|L_t(U_t)|$, because $G_t$ still maximizes over all future vertex
blocks and may use the fixed future vertices chosen above.  So by
Lemma~\ref{lem:one-block-uniform} 
\[
        \Ex\left[\left(\log(G_{t-1}/G_t)\right)_+
        \Bigm| U_1,\ldots,U_{t-1}\right]\le1.
\]
Summing over $t$ gives $\Ex\left[\left(\log(D_{\max}/|D(U)|)\right)_+\right]
        \le r$, and the result now follows by applying Markov's inequality.
\end{proof}

We now prove Lemma \ref{lem:second-terminal}
\begin{proof}{(Lemma \ref{lem:second-terminal}).}
After sampling $U$, consider the rounding algorithm that greedily rounds $U_t$ to  vertex, for each block one at a time. Let $D_t$ denote the determinant after time $t$.
Suppose
we have already rounding $U_1,\ldots,U_{t-1}$ to vertices $\zeta_1,\ldots \zeta_{t-1}$ and are currently rounding part $t$ (the blocks $t+1,\ldots, r$ are still at their sampled values). 
By multilinearity, $
        D_{t-1}=\sum_{i\in P_a}U_{t,i}c_i$,
for some $c_i$'s (independent of $U_t$). 

The algorithm tries all $i\in P_a$ and chooses the one maximizing
$|c_i|$, so that
        $|D_{t}|=\max_i|c_i|$.
As $|U_{t,i}|\le1$,
\[
        |D_{t-1}|
        \le\sum_{i\in P_t}|U_{t,i}||c_i|
        \le n_t\max_i|c_i|
        =n_t|D_{t}|.
\]
Thus rounding block $t$ loses at most a factor $m_t$, deterministically.
Multiplying over all blocks,
the final vertex $\zeta$ satisfies
        $|D(\zeta)|\ge |D(U)|\prod_{a=1}^rn_a^{-1}$.
        \end{proof}

%% file: small-k-no-highlights.tex
\section{The case of $k\le d$}
\label{sec:k-le-d}

We now explain the modifications needed when the number of chosen vectors is
$k\le d$. As before, we assume that there are $k$ parts, each with rank $1$.   When $k<d$, the matrix $\sum_{i\in B}v_iv_i^\top$ is singular,
so its determinant is zero.  The right objective is instead the squared
$k$-dimensional volume
\[
\operatorname{vol}_k(B)^2:=\det(V_B^\top V_B),
\]
where $V_B$ is the $d\times k$ matrix whose columns are the selected vectors.  If
$\operatorname{sym}_s(M)$ denotes the $s$th elementary symmetric polynomial of the
eigenvalues of a positive semidefinite matrix $M$, then
\[
\det(V_B^\top V_B)
=
\operatorname{sym}_k\Big(\sum_{i\in B}v_iv_i^\top\Big).
\]
Indeed, by Cauchy--Binet, both sides are the sum of squares of all $k\times k$
minors of $V_B$.

The proof for $k\le d$ follows the determinant case with one systematic change:
in a residual component with $s$ remaining parts, the determinant objective is
replaced by the degree-$s$ elementary-symmetric objective.  The ambient dimension
of the residual vectors may be larger than $s$, but the degree of the objective is
$s$.  All losses below are proportional to this residual degree, and hence sum to
$O(k)$ over all residual components.

\subsection{The $\operatorname{sym}_s$ scaling relaxation}

As before, for an instance with parts $P_1,\ldots,P_s$, define
$A_a(x):=\sum_{i\in P_a}x_iv_iv_i^\top$, with
$x(P_a)=1$.

We write $G_s$ for the $\operatorname{sym}_s$ analogue of the scaling relaxation:
\begin{equation}
\label{eq:sym-s-relaxation}
G_s(x)
:=
\inf_{z\in \mathbf 1^\perp}
\log \operatorname{sym}_s\Big(\sum_{a=1}^s e^{z_a}A_a(x)\Big),
\end{equation}
where $\mathbf 1^\perp=\{z\in\mathbb R^s:\sum_a z_a=0\}$. 
infimum is over positive scalars $y_a=e^{z_a}$ with $\prod_a y_a=1$.

We shall repeatedly use the following Cauchy--Binet expansion.  If
$w_i=e^{z_{a(i)}/2}v_i$, where $a(i)$ is the part containing $i$, then
\begin{equation}
\label{eq:sym-cb}
\operatorname{sym}_s\Big(\sum_i x_iw_iw_i^\top\Big)
=
\sum_{T:|T|=s}\det(W_T^\top W_T)\prod_{i\in T}x_i .
\end{equation}
Here $W_T$ is the matrix with columns $w_i$, $i\in T$.  If the columns in $T$ are
linearly dependent, the determinant coefficient is zero. We denote $M:=\sum_i x_iw_iw_i^\top$.

The following basic facts are the direct analogues of the determinant case.

\begin{lemma}[Basic properties of $G_s$]
\label{lem:sym-basic}
The relaxation \eqref{eq:sym-s-relaxation} is exact on integral transversals:
if $B=\{i_a\in P_a:a\in[s]\}$, then
\[
G_s(1_B)=\log\det(V_B^\top V_B).
\]
Moreover, $G_s$ is concave in $x$.  Finally, at any nondegenerate point $x$ where
the inner minimizer $z(x)$ exists, the $\operatorname{sym}_s$-leverage scores
\[
\ell_i:=D\log \operatorname{sym}_s(M)[w_iw_i^\top]
\]
satisfy the partwise balance identity
$
\sum_{i\in P_a}x_i\ell_i=1$ for every $a\in[s].
$
\end{lemma}

\begin{proof}
Exactness follows because product-one scalings multiply every integral
transversal by total factor one.  Concavity follows from the concavity of
$M\mapsto\log\operatorname{sym}_s(M)$ on the positive semidefinite cone, the
affineness of $\sum_a e^{z_a}A_a(x)$ in $x$ for fixed $z$, and preservation of
concavity under taking the infimum over $z$.

For the balance identity, differentiate the inner objective with respect to
$z_a$ at the minimizer.  The derivative is
\[
D\log\operatorname{sym}_s(M)\Big(\sum_{i\in P_a}x_iw_iw_i^\top\Big)
=
\sum_{i\in P_a}x_i\ell_i.
\]
At a constrained minimizer over $\mathbf 1^\perp$, these $s$ derivatives are all
equal.  Their sum is $s$, since $\log\operatorname{sym}_s$ is homogeneous of
degree $s$.  Hence each derivative equals $1$.
\end{proof}

The degeneracy and splitting discussion from Section~\ref{sec:nondegenerate-supports} is unchanged. In particular, exact singleton contractions and exact
tight-set splits have the same geometric meaning as before: if a set $S$ of
parts spans a subspace of dimension $|S|$, the chosen vectors from $S$
contribute the base volume in that subspace, and the remaining chosen vectors
contribute the height after orthogonal projection.

\subsection{The calculus via the Cauchy--Binet distribution}
We now give the analogue of Section~\ref{subsec:probabilistic-derivative-dictionary}.  

Fix
a nondegenerate point $x$ and balance at its inner minimizing scaling. Let us denote
$Z_x:=\operatorname{sym}_s(M)$.
Define a probability distribution $\mu_x$ on $s$-subsets by
\begin{equation}
\label{eq:sym-cb-distribution}
\mu_x(T)
:=
\frac{\det(W_T^\top W_T)\prod_{i\in T}x_i}{Z_x}.
\end{equation}
As before, for $T\sim\mu_x$, let $
R_a(T):=|T\cap P_a|$ and for a tangent direction $\delta$, set
\[
\gamma_i:=\frac{\delta_i}{x_i},
\qquad
\sigma_\delta(T):=\sum_{i\in T}\gamma_i,
\qquad
A_{2,\delta}(T):=\sum_{i\in T}\gamma_i^2.
\]

\begin{lemma}
\label{lem:sym-derivative-identities}
At a nondegenerate point $x$, the following identities hold:
\begin{enumerate}
\item $x_i\ell_i=\Pr_{\mu_x}[i\in T]$ for every coordinate $i$, and
$\mathbb E_{\mu_x}R_a=1$ for every part $a\in[s]$.
\item $DG_s(x)[\delta]=\mathbb E_{\mu_x}\sigma_\delta$.  For fixed $z$,
$
D_{xx}^2\log\operatorname{sym}_s(M(x))[\delta,\delta]
=
\operatorname{Var}_{\mu_x}(\sigma_\delta)-\mathbb E_{\mu_x}A_{2,\delta}.$
\item The mixed derivatives and the $z$-Hessian are
$
\eta_a(\delta):=\operatorname{Cov}_{\mu_x}(R_a,\sigma_\delta),
$ and 
$\Gamma_{ab}:=\operatorname{Cov}_{\mu_x}(R_a,R_b)$.
\item Define $
\mathcal B_x(\delta)
:=
\mathbb E_{\mu_x}A_{2,\delta}-\operatorname{Var}_{\mu_x}(\sigma_\delta).$
Then $\mathcal B_x(e_i)=\ell_i^2$, and
\[
D^2G_s(x)[\delta,\delta]
=
-\mathcal B_x(\delta)-\eta(\delta)^\top\Gamma^{-1}_{\mathbf 1^\perp}\eta(\delta).
\]
In particular, if $\eta(\delta)=0$, then
$D^2G_s(x)[\delta,\delta]=-\mathcal B_x(\delta)$.
\end{enumerate}
\end{lemma}

\begin{proof}
Differentiating \eqref{eq:sym-cb} with respect to $x_i$ selects exactly the
terms containing $i$.  This gives
\[
x_i\ell_i
=
\frac{\sum_{T\ni i}\det(W_T^\top W_T)\prod_{j\in T}x_j}{Z_x}
=
\Pr[i\in T].
\]
Summing this identity over $i\in P_a$ and using Lemma~\ref{lem:sym-basic} gives
$\mathbb E R_a=1$.

For the directional derivative, write
\[
\prod_{i\in T}(x_i+t\delta_i)
=
\prod_{i\in T}x_i\prod_{i\in T}(1+t\gamma_i).
\]
Thus
\[
\log\operatorname{sym}_s(M(x+t\delta))
=
\log Z_x+
\log\mathbb E_{\mu_x}\prod_{i\in T}(1+t\gamma_i).
\]
Differentiating once at $t=0$ gives the identity for fixed $z$; by the
envelope theorem, this is also $DG_s(x)[\delta]$ at the inner minimizer.  Differentiating twice at fixed $z$ gives the second-derivative identity, because
\[
\left.\frac{d^2}{dt^2}\right|_{t=0}\prod_{i\in T}(1+t\gamma_i)
=
\sigma_\delta(T)^2-A_{2,\delta}(T).
\]
The positive semidefiniteness of $\mathcal B_x$ is the concavity of
$\log\operatorname{sym}_s$ in the matrix direction.

For the diagonal identity, take $\delta=e_i$.  Then
$\sigma_\delta=x_i^{-1}{\bf 1}_{i\in T}$ and
$A_{2,\delta}=x_i^{-2}{\bf 1}_{i\in T}$.  Using
$\Pr[i\in T]=x_i\ell_i$ gives
$
\mathbb E A_{2,\delta}-\operatorname{Var}(\sigma_\delta)
=
\ell_i^2.$

The formulas for $\eta$ and $\Gamma$ are obtained by differentiating the same
log-sum-exp in the $x$ and $z$ directions; the $z$ variables insert the part counts
$R_a$.  The calculation for the optimized Hessian is exactly the
same as in before, with $\mathcal B_x$ replacing
$\operatorname{tr}(C(\delta)^2)$.
\end{proof}

\subsection{Compression Argument}

We now describe the compression argument.  We use the same capped potential as in
Section~\ref{sec:continuous-compression} and set
$\Phi(x)=\sum_i p(x_i)$, and 
$L(x)=G_s(x)+\Phi(x)$ on a rank-$s$ residual component.  Since $\sum_i x_i=s$, the
range of $\Phi$ is $
-s\le \Phi(x)\le 0.$
\begin{lemma}[First-order progress and leverage control for $G_s$]
\label{lem:sym-first-order-progress}
For any nondegenerate point $x$ in a rank-$s$ residual component, either there is
a unit tangent direction $\delta$ with $DL(x)[\delta]>1/\sqrt2$, or else
$\ell_i\le 3$ for every active coordinate $i$.
\end{lemma}

\begin{proof}
We the balance identity from Lemma~\ref{lem:sym-basic}.  If two active
coordinates in the same part have $L$-gradients differing by more than $1$, their
difference gives the desired first-order direction.  Otherwise, the gradients in
each part differ by at most $1$.  The $x$-weighted average of the gradients in
part $P_a$ is
\[
\sum_{i\in P_a}x_i(\ell_i+p'(x_i))
=
1+
\sum_{i\in P_a}x_ip'(x_i)
\le 1.
\]
Since $p'\in[-1,0]$, every active coordinate has $\ell_i\le 3$.
\end{proof}

The finite rebalancing lemma from Section~\ref{sec:finite-dynamics} also transfers
verbatim, with $F$ replaced by $G_s$.  Thus a small rebalancing step gives the
same leverage bound $\ell_i\le3$, while a large rebalancing step gives a fixed
increase in $L$.

Let
$I_{\rm sm}(x):=\{i:\tau<x_i\le 1/(2H)\}$
be the set of small coordinates above the deletion floor.  On the coordinate
space supported on $I_{\rm sm}(x)$, the trace of the quadratic form
$\mathcal B_x$ is
\[
\sum_{i\in I_{\rm sm}(x)}\mathcal B_x(e_i)
=
\sum_{i\in I_{\rm sm}(x)}\ell_i^2
\le
9|I_{\rm sm}(x)|.
\]
The part constraints $\delta(P_a)=0$ and the null-Schur constraints
$\eta(\delta)=0$ impose at most $s+(s-1)$ independent linear equations.  Therefore,
if $|I_{\rm sm}(x)|>4s$, there is a unit vector $\delta$, supported on
$I_{\rm sm}(x)$, satisfying $\delta(P_a)=0$ for every part, $\eta(\delta)=0$, and
$\mathcal B_x(\delta)\le18$.  Choosing the sign so that $DL(x)[\delta]\ge0$, the
positive curvature $H$ of the potential on small coordinate
The finite null-Schur estimate also transfers directly.

\begin{lemma}
\label{lem:sym-finite-null-schur}
Let $x$ be nondegenerate, and let $\delta$ be a feasible direction satisfying
$\delta(P_a)=0$ for every part and $\eta(\delta)=0$.  Let
$K_0:=\max_{i\in\operatorname{supp}(x)} |\delta_i|/x_i$.  Then, for
$|t|\le 1/(4sK_0)$,
\[
G_s(x+t\delta)
\ge
G_s(x)+tDG_s(x)[\delta]
-\frac{t^2}{2}\mathcal B_x(\delta)
-O(s^3K_0^3|t|^3).
\]
\end{lemma}

\begin{proof}
The proof is the same balanced-certificate proof as Lemma~\ref{lem:finite-null-schur-estimate-balanced}.  The only inputs used there are the Cauchy--Binet distribution, the balanced identity $\mathbb E R=\mathbf 1$, the covariance identity
$\operatorname{Cov}(R,\sigma_\delta)=\eta(\delta)$, and the bound
$|\sigma_\delta(T)|\le sK_0$.  These are exactly the identities supplied by
Lemma~\ref{lem:sym-derivative-identities}.  No property of the ambient dimension is
used.
\end{proof}

The deletion step transfers as well.

\begin{lemma}[Small-coordinate deletion for $G_s$]
\label{lem:sym-small-deletion}
Suppose $x$ is nondegenerate, $x_i=\tau$, and $\ell_i\le3$.  Let $\bar x$ be
obtained by deleting coordinate $i$ and renormalizing its part.  Then
\[
G_s(\bar x)\ge G_s(x)-O(s^2\tau),
\qquad
L(\bar x)\ge L(x)-O(s^2\tau).
\]
In particular, for the choice $\tau=1/\operatorname{poly}(n,d)$ used in
Section~\ref{sec:finite-dynamics}, the total deletion loss over all deletions is
negligible compared with the final $O(k)$ log-loss.
\end{lemma}

\begin{proof}
By Lemma~\ref{lem:sym-derivative-identities}, the Cauchy--Binet probability mass
of terms containing $i$ is $x_i\ell_i\le3\tau$.  Removing $i$ therefore deletes
only $O(\tau)$ coefficient mass at the current balanced scaling.  The same
balanced-certificate argument used in Lemma~\ref{lem:finite-null-schur-estimate-balanced}, with the degree $s$ replacing the determinant dimension, shows that the
optimized scaling value decreases by at most $O(s^2\tau)$.  The renormalization of
the surviving coordinates in the part moves total mass $\tau$, and since
$|p'|\le1$, the potential changes by at most $O(\tau)$.  This proves the claim.
\end{proof}

Combining these observations  gives the precise analogue of the polynomial-time compression
result.
Finally, the terminal rounding theorem of Ebrahimi--Straszak--Vishnoi, applied to
the compressed support, loses only another $\exp(O(k))$ factor because the total
support is $O(k)$.  Multiplying the residual determinant by the exact factors
accumulated during contractions gives the final result.